\documentclass[final]{IEEEtran}
\usepackage{color}
\usepackage{csquotes}
\usepackage{enumitem}
\usepackage{float}
\usepackage{multirow}

\floatstyle{ruled}
\newfloat{algorithm}{tbp}{loa}
\providecommand{\algorithmname}{Algorithm}
\floatname{algorithm}{\protect\algorithmname}
\usepackage{hyperref}
\usepackage{amsopn}

\newcommand{\nn}{\nonumber}

\DeclareMathOperator*{\argmax}{arg\,max}

\makeatletter
\newcommand{\manuallabel}[2]{\def\@currentlabel{#2}\label{#1}}
\makeatother
\usepackage{amsmath}
\usepackage{amssymb}
\usepackage{amsthm}
\usepackage{bbm}
\usepackage{tikz}
\usepackage{mathrsfs}
\usepackage{pgfplots}
\pgfplotsset{compat=1.14}
\usetikzlibrary{arrows}
\usetikzlibrary{arrows.meta}

\newtheorem{theorem}{Theorem}
\newtheorem{lemma}[theorem]{Lemma}

\newtheorem{dfn}[theorem]{Definition}

\newtheorem{cor}[theorem]{Corollary}

\newtheorem{assump}[theorem]{Assumption}

\newtheorem{fact}[theorem]{Fact}

\newcommand{\Ptilde}{\widetilde{P}}

\newcommand{\p}{\mathbb{P}}
\newcommand{\e}{\mathbb{E}}

\newcommand{\calX}{\mathcal{X}}
\newcommand{\calY}{\mathcal{Y}}
\newcommand{\calZ}{\mathcal{Z}}
\newcommand{\calW}{\mathcal{W}}

\newcommand{\z}{\vec{w}}
\newcommand{\w}{\vec{w}}

\newcommand{\dc}{d_{\rm C}}

\newcommand{\yk}{y_{1\ldots k}}
\newcommand{\yn}{y_{1\ldots n}}
\newcommand{\wn}{w_{1\ldots n}}
\newcommand{\dq}[1]{D(Q_s\|Q_{#1})}
\newcommand{\dqs}[1]{D(Q_{1-s}\|Q_{#1})}
\newcommand{\dt}[1]{D(Q_t\|Q_{#1})}
\newcommand{\dqt}[1]{D(Q_{1-t}\|Q_{#1})}

\newcommand{\pyxi}[1]{P(y|\vec{x}_{#1}^{(i)})}
\newcommand{\pyxii}[1]{P(y_i|\vec{x}_{#1}^{(i)})}

\title{Optimal 1-bit Error Exponent for 2-hop \\ Relaying with Binary-Input Channels}
\author{Yan Hao Ling and Jonathan Scarlett\thanks{The authors are with the  Department of Computer Science, School of Computing, National University of Singapore (NUS). J.~Scarlett is also with the Department of Mathematics, NUS, and the Institute of Data  Science, NUS. Emails: \url{lingyh@nus.edu.sg};  \url{scarlett@comp.nus.edu.sg}}}

\begin{document}
\maketitle

\begin{abstract}
    In this paper, we study the problem of relaying a single bit over a tandem of binary-input channels, with the goal of attaining the highest possible error exponent in the exponentially decaying error probability.  Our previous work gave an exact characterization of the best possible error exponent in various special cases, including when the two channels are identical, but the general case was left as an open problem.  We resolve this open problem by deriving a new converse bound that matches our existing achievability bound.
\end{abstract}

\section{Introduction}

The problem of relaying a single bit over tandem of channels has recently gained attention from a number of perspectives, with Huleihel, Polyanskiy, and Shayevitz adopting conventional relaying-based terminology and connecting the problem to that of many-hop relaying (information velocity) \cite{onebit}, and Jog and Loh motivating the same problem from the perspective of multi-agent teaching and learning \cite{jog2020teaching} (see also earlier works such as \cite{jadbabaie2013information,molavi2017foundations}).  These works focused on a tandem of binary symmetric channels (BSCs), and our subsequent work \cite{teachlearn} showed that in this case, the optimal 2-hop exponent is identical to the 1-hop exponent of the weaker channel.

In \cite{teachlearn}, we additionally considered general binary-input discrete memoryless channels, which we believe to be a natural next step after BSCs given that channel symmetry is not always to be expected in practice (whether in the context of relay channels or teaching/learning).  Specifically, we gave achievability and converse bounds that match in certain special cases, with a notable example being that the ``1-hop = 2-hop'' finding remains true when the two channels are identical.  However, we also gave examples where such a result is provably false, and where gaps remained in our bounds.

In this paper, we close all of the remaining gaps from \cite{teachlearn} for the case of binary-input discrete memoryless channels.  To do so, we derive an improved converse that matches the achievability result therein.  This converse is proved via entirely different techniques to \cite{teachlearn}, where we used a genie argument and a reduction to communication with feedback that appears to be inherently loose.

\subsection{Problem Setup}

We first formalize the model, which is depicted in Figure \ref{fig:setup}. There are three agents: an encoder, relay, and decoder. The ``encoder $\to$ relay'' channel is a discrete memoryless channel (DMC) $P$ and the ``relay $\to$ decoder'' is another DMC $Q$.  The input alphabets and output alphabets of $P$ are denoted by $\calX$ and $\calY$ respectively, and the input alphabets and output alphabets of $Q$ are denoted by $\calW$ and $\calZ$ respectively.  We limit our attention to binary-input channels, meaning that $\calX = \calW = \{0,1\}$.
 The unknown message of interest is a random variable $\Theta$ drawn uniformly from $\{0,1\}$, i.e., the message is one bit.

     \begin{figure*}[!t]
    \centering
    \begin{tikzpicture}
\draw (0.75,1.25) node {$\Theta$};
\draw (3.25,1.25) node {$X$};
\draw (5.75,1.25) node {$Y$};
\draw (8.25,1.25) node {$W$};
\draw (10.75,1.25) node {$Z$};
\draw (13.25,1.25) node {$\hat{\Theta}$};
\draw[->] (0.5,1) -- (1,1);
\draw[thick] (1,0.5) -- (1,1.5) -- (3,1.5) -- (3,0.5) -- (1,0.5);
\node at (2, 1) {Encoder};
\draw[->] (3,1) -- (3.5,1);
\draw[thick] (3.5,0.5) -- (3.5,1.5) -- (5.5,1.5) -- (5.5,0.5) -- (3.5,0.5);
\node at (4.5, 1) {$P$};
\draw[->] (5.5,1) -- (6,1);
\draw[thick] (6,0.5) -- (6,1.5) -- (8,1.5) -- (8,0.5) -- (6,0.5);
\node at (7, 1) {Relay};
\draw[->] (8,1) -- (8.5,1);
\draw[thick] (8.5,0.5) -- (8.5,1.5) -- (10.5,1.5) -- (10.5,0.5) -- (8.5,0.5);
\node at (9.5, 1) {$Q$};
\draw[->] (10.5,1) -- (11,1);
\draw[thick] (11,0.5) -- (11,1.5) -- (13,1.5) -- (13,0.5) -- (11,0.5);
\node at (12, 1) {Decoder};
\draw[->] (13,1) -- (13.5,1);
\end{tikzpicture}
    \caption{Illustration of our problem setup.}
    \label{fig:setup}
\end{figure*}
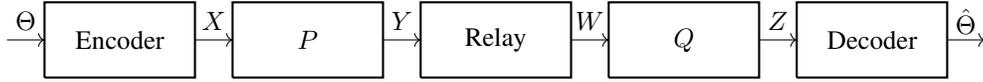

At time step $i \in \{1,\dotsc, n\}$, the following occurs (simultaneously):
    \begin{itemize} 
        \item The encoder transmits to the relay via one use of $P$. Let $X_i$ denote the input from the encoder and $Y_i$ denote the output received by the relay.
        \item The relay transmits to the decoder via one use of $Q$. Let $W_i$ denote the input from the relay and $Z_i$ denote the output received by the decoder.
    \end{itemize}
    Importantly, $W_i$ must only depend on $Y_1, \ldots, Y_{i-1}$ and not $Y_i,\ldots,Y_n$, i.e., the relay is not allowed to use information from the future.  At time $n$, having received $Z_1, \ldots, Z_n$, the decoder forms an estimate of $\Theta$, which we denote by $\hat{\Theta}_n$ (or sometimes simply $\hat{\Theta}$). 
    
    Let $P_e(n, P, Q) = \p(\hat{\Theta}_n \ne \Theta)$ be the error probability after $n$ time steps.  Then, the 2-hop error exponent is defined as
    \begin{equation}
        E(P,Q) = \sup_{{\rm protocols}} \liminf_{n\rightarrow \infty} \left\{ -\frac1n \log P_e(n,P,Q)\right\},
        \label{learning_rate}
    \end{equation}
    where ``protocols'' refers to the set of all possible designs of the encoder, relay, and decoder.

    We will often be interested in cases where the 2-hop error exponent is equal to a 1-hop exponent, where the latter is defined as follows: In a 1-hop setting where the encoder transmits directly to the decoder via the channel $P$ (respectively, $Q$), we define the corresponding highest possible error exponent as $E(P)$ (respectively, $E(Q)$).

    \subsection{Statement of Main Result}

    In our earlier work \cite{teachlearn}, we established the following achievability result.
    
    \begin{theorem}
        {\em (\cite[Thm.~2]{teachlearn})}  Let $P, Q$ be binary input DMCs, and let $P_0, P_1, Q_0, Q_1$ represent the conditional output distributions of $P$ and $Q$ for inputs $0,1$ respectively. Then\footnote{Having two inner ``max'' operations in \eqref{eq:achievability} is redundant, because if we replace one of them (but not the other) by their first argument, the overall quantity is unchanged.  However, this form will be slightly more convenient to work with.}
        \begin{align}
            E(P,Q) \geq &\max_{0\leq s \leq 1} \bigg(\min\Big( \nn \\
            &\max\big(\dc(P_0, P_1, s), \dc(P_0, P_1, 1-s)\big), \nn \\ &\max\big(\dc(Q_0, Q_1, s), \dc(Q_0, Q_1, 1-s)\big)\Big)\bigg), \label{eq:achievability}
        \end{align}
        where $\dc(P_0.P_1,s) = -\log \sum_{x \in \calX} P_0(x)^{1-s}P_1(x)^{s}$ for $s \in (0,1)$, and the endpoints $s \in \{0,1\}$ are defined via continuity (and similarly for the other $\dc$ terms in \eqref{eq:achievability}).
        \label{thm:teachlearn}
    \end{theorem}

    We showed in \cite{teachlearn} that Theorem \ref{thm:teachlearn} is tight in certain special cases, including when $P=Q$ or when $Q$ is a binary symmetric channel.  However, its general tightness was left as an open problem.  In this paper, we resolve this open problem by showing that Theorem \ref{thm:teachlearn} is tight with no further assumptions.
    
    \begin{theorem}
        The inequality in Theorem \ref{thm:teachlearn} is tight for any binary-input channels $P,Q$, i.e.
        \begin{align}
            E(P,Q) = &\max_{0\leq s \leq 1} \bigg(\min\Big( \nn \\
            &\max\big(\dc(P_0, P_1, s), \dc(P_0, P_1, 1-s)\big), \nn \\
            & \max\big(\dc(Q_0, Q_1, s), \dc(Q_0, Q_1, 1-s)\big)\Big)\bigg). 
        \label{eq:2hop_optimal}
        \end{align}
        \label{thm:binary_optimal}
    \end{theorem}

    Note that if $P_0$ and $P_1$ have disjoint support, then it is possible for the encoder to convey a single bit to the relay in a single time step, after which the relay can then execute any 1-hop optimal protocol.  Similarly, if $Q_0$ and $Q_1$ have disjoint support, then the relay can wait until the last time step, make the best guess of $\Theta$, and then transmit 1 bit to the decoder with absolute certainty.  In such scenarios, the right-hand side of \eqref{eq:2hop_optimal} reduces to a 1-hop exponent $E(P)$ or $E(Q)$, and the result immediately from known 1-hop results \cite{berlekampI}.  We therefore exclude the above scenarios throughout the rest of these paper; it is then easily checked that $\dc(P_0, P_1, s)$ and $\dc(Q_0, Q_1, s)$ are finite for all $0\leq s \leq 1$.

    We will in fact show a more precise error bound (stated below), which immediately implies the upper bound on $E(P,Q)$ required to establish Theorem \ref{thm:binary_optimal}.
    Before stating the more precise bound, we introduce the following definition, which will play a crucial role in our analysis.
    
    \begin{dfn}
        $p_{\min}$ denotes the smallest non-zero transition probability among the channels $P$ and $Q$.
    \end{dfn}

    Note that we always have $p_{\min} > 0$, and this definition does not preclude the possibility of zero-probability transitions in $P$ and/or $Q$.
    
    \begin{theorem}
        Under the conditions of Theorem \ref{thm:teachlearn}, let $p_{e,0}$ and $p_{e,1}$ be the error probability conditioned on $\Theta=0$ and $\Theta=1$ respectively, and let $E^*$ denote the right-hand side of \eqref{eq:achievability}. Then, regardless of the choice of encoder, relay, and decoder, we have
    \begin{equation}
        -\log (p_{e,0} + p_{e,1}) \leq nE^* + (\sqrt{2n} + 4)\log \frac{1}{p_{\min}} + \log 8.
    \end{equation}
    \label{thm:binary_optimal_exact}
    \end{theorem}
    
    The rest of the paper will build towards the proof of Theorem \ref{thm:binary_optimal_exact}.  We introduce some preliminary notation and results in Section \ref{sec:prelim}, provide some useful 1-hop communication results in Section \ref{sec:1-hop}, give an overview of our main proof in Section \ref{sec:overview}, identify some simplifying conditions that can be assumed in Section \ref{sec:reduction}, and complete the proof under those conditions in Section \ref{sec:main}.

\subsection{Other Related Work}

Two prior works are by far the most related, namely, our previous work \cite{teachlearn} in which we provided the achievability result stated in Theorem \ref{thm:teachlearn}, and the early work of Shannon, Gallager, and Berlekamp \cite{berlekampI} on the optimal error exponent in the 1-hop setting, from which we borrow several technical tools.  Other works are less directly related, so we only provide a brief summary.

As mentioned in the introduction, the initial works \cite{onebit,jog2020teaching} focused on binary symmetric channels.  They studied various relaying strategies such as direct forwarding and forwarding of the best guess so far, and found that none of these dominate one another in general.  Our work \cite{teachlearn} solved the BSC case and studied more general binary-input channels, but only fully solved it in certain special cases.  For comparison to our main result, we outline two converse bounds that were given in \cite{teachlearn}:
\begin{itemize}
    \item The optimal error exponent is trivially upper bounded by the worse of the two 1-hop exponents:
    \begin{align}
        &E(P,Q) \le \min\big\{ E(P), E(Q) \big\} \nn \\ 
        &= \min\bigg( \max_{0 \le s \le 1} \dc(P_0,P_1,s), \max_{0 \le s \le 1} \dc(Q_0,Q_1,s) \bigg). \label{eq:trivial_converse}
    \end{align}
    Observe that if the order of the max (with respect to $s$) and min (of two terms) were swapped in \eqref{eq:achievability}, then the result would match \eqref{eq:trivial_converse}.  This implies that \eqref{eq:trivial_converse} is tight when the same $s$ value maximizes both terms, e.g., when $P=Q$.  However, in general, swapping the min-max order can change the value.
    \item An improved converse was given in \cite{teachlearn} based on a genie argument that gives the decoder direct access to the relay's output for the first $\gamma n$ time steps, and gives the relay direct access to $\Theta$ for the remaining $(1-\gamma)n$ time steps.  This reduces the problem to a point-to-point channel \emph{with feedback}, namely, $\gamma n$ uses of $P$, followed by full feedback of all outputs so far, followed by $(1-\gamma)n$ uses of $Q$.  It is shown in \cite{teachlearn} that when $P$ is a BSC and $Q$ is a Z-channel, there is still a gap between the resulting converse and the exponent in Theorem \ref{thm:teachlearn}.
\end{itemize}

Other directions in this line of works have included multi-bit relaying \cite{multibit} (where the transmitted message
$\Theta$ is no longer binary) and more general graphical network structures \cite{multirelay,optbiso,maxflow}.  As discussed in \cite{onebit}, the problem that we study also has connections to many-hop settings and information velocity, which was further studied in \cite{infovelocity,packeterasure,domanovitz2023information} (see also the much earlier work \cite{schulman_1994}).

The problem of characterizing the \emph{capacity} of relay channels is a substantially different one, for which we refer the reader to \cite{cover_thomas,gamalkim} for an overview.  Similarly, the study of \emph{positive-rate} error exponents for relay channels (e.g., see \cite{bradford2012error, highrates, ogbe2019optimal}) is largely distinct, relying on different techniques such as random coding.

\section{Preliminaries} \label{sec:prelim}

\subsection{Further Definitions and Notation}

The encoder, upon seeing $\Theta=0$ or $\Theta=1$, sends out a length-$n$ codeword, which we will denote as $\vec{x}_0$ or $\vec{x}_1$ respectively. We let $p_{\min} > 0$ denote the smallest nonzero transition probability of the channels $P$ and $Q$ (while still allowing zero-probability transitions in both channels).

\begin{dfn}
For any string $\w$ with symbols in $\{0,1\}$, we will use $n_0(\w)$ (and similarly $n_1(\w)$) to represent the number of 0s (and similarly number of 1s) in the string $\w$.
\label{dfn:n0_n1}
\end{dfn}

The following ``tilted'' distributions are widely used in the study of 1-hop error exponents \cite{berlekampI}, and will also be used frequently in our work.

\begin{dfn} \label{def:tilted}
For all $0<s<1$, define the distributions $P_s$ and $Q_s$ by
\begin{equation}
    P_s(y) = \frac{P_0(y)^{1-s}P_1(y)^s}{\sum_{y'} P_0(y')^{1-s}P_1(y')^s}
\end{equation}
and
\begin{equation}
    Q_s(z) = \frac{Q_0(z)^{1-s}Q_1(z)^s}{\sum_{z'} Q_0(z')^{1-s}Q_1(z')^s}.
    \label{eq:def_qs}
\end{equation}
\end{dfn}


We will use $\yk$ as a shorthand for $(y_1, y_2,\ldots, y_k)$, and similarly for $\yn$ and other sequences such as $\wn$.  We will use notation such as $P(\yk)$ as a shorthand for $\prod_{i=1}^k P(y_i)$ (and similarly for $Q$), or alternatively $P^k(\yk)$ when we want to make the dependence on the length explicit.
For all $\yk$, we define
\begin{align}
    p_{e,0}(\yk) = &~\p(\hat{\Theta} = 1 | \Theta=0 \hbox{ and the relay} \nn \\ &\hbox{ receives $\yk$ in the first $k$ time steps}) 
    \label{eq:pe0_def}
\end{align}
\begin{align}
    p_{e,1}(\yk) = &~\p(\hat{\Theta} = 0 | \Theta=1 \hbox{ and the relay} \nn \\ &\hbox{receives $\yk$ in the first $k$ time steps}),
    \label{eq:pe1_def}
\end{align}
which simplify to $p_{e,0}$ and $p_{e,1}$ (used in Theorem \ref{thm:binary_optimal_exact}) when $k=0$.

Even though our problem setup allows the use of randomized protocols, it is sufficient to prove the error bound for deterministic protocols. This is because any randomized protocol can be viewed as a distribution over deterministic protocols, and therefore no randomized protocol can attain a smaller average error probability than the best deterministic protocol.

Under the assumption that the encoder is deterministic, the two possible transmitted codewords are fixed sequences, say $\vec{x}_0$ and $\vec{x}_1$ for the cases $\Theta=0$ and $\Theta=1$ respectively.  
Similarly, under the assumption that the relay behaves deterministically, the first $k$ symbols sent out by the relay are given by a fixed function of $\yk$. We can therefore define\footnote{In fact, $\yk$ determines the first $k+1$ symbols sent by the relay, but for our purposes, using this definition with the first $k$ symbols is sufficient, and leads to slightly simpler notation.}
\begin{align}
    \w(\yk) = &\hbox{ The first $k$ symbols sent out } \nn \\ & \hbox{ by the relay after receiving $\yk$}. \label{eq:w_def}
\end{align}
When considering a partial received string $\yk$ at the relay, and hence fixed $\w(\yk)$, it will often be useful to consider the \emph{best case scenario} for the remaining $n-k$ entries of $\wn$ (since any given protocol will have performance no better than this best case, and we are seeking a converse result).  To formalise this, let $D_0$ and $D_1$ be decoding regions over the output of $Q^n$ (i.e., the decoder estimates $\hat{\Theta} = \nu$ when the output is in $D_{\nu}$).  
Then, we let $\w_1(\yk)$ be the extension of $\w(\yk)$ to a string $\wn$ of length $n$ which minimizes
\begin{equation}
    \p_{\vec{z} \sim Q^n(\cdot|\wn)}(\vec{z} \in D_0). \label{eq:w1_def}
\end{equation}
Observe that
\begin{equation}
    p_{e,1}(\yk) = \e (\p_{\vec{z} \sim Q^n(\cdot|\wn)}(\vec{z} \in D_0)),
\end{equation}
where the expectation is taken over some suitable distribution on $\wn$, each realization of which is an extension of $\w(\yk)$. Therefore, the definition of $\w_1(\yk)$ gives
\begin{equation}
    \p_{\vec{z} \sim Q^n(\w_1(\yk))}(\vec{z} \in D_0) \leq p_{e,1}(\yk).
    \label{eq:w1pe1}
\end{equation}
Similarly let $\w_0(\yk)$ be the extension of $\w(\yk)$ to a string of length $n$, $\wn$ which minimizes
\begin{equation}
    \p_{\vec{z} \sim Q^n(\cdot|\wn)}(\vec{z} \in D_1),
\end{equation}
again noting that
\begin{equation}
    \p_{\vec{z} \sim Q^n(\cdot|\w_0(\yk))}(\vec{z} \in D_1) \leq p_{e,0}(\yk).
    \label{eq:w0pe0}
\end{equation}
If $k=0$, then $\yk$ is the ``empty string'', and we will denote the corresponding strings $\w_0(\yk)$ and $\w_1(\yk)$ by $\w_0(\emptyset)$ and $\w_1(\emptyset)$.

The following observation relating the performance given $\yk$ vs.~$y_{1\ldots k+1}$ will be useful.  Here and throughout the rest of the paper, we let $(\cdot)^{(i)}$ denote the $i$-th element of a sequence.

\begin{lemma} \label{lem:pmin_bound}
    Fix any sequence $y_{1\dotsc k+1}$, and let $\vec{x}_1 = (\vec{x}_1^{(1)},\dotsc,\vec{x}_1^{(n)})$ be the codeword sent out by the encoder when $\Theta=1$.  If $P(y_{k+1}|\vec{x}_1^{(k+1)})>0$, then
    \begin{equation}
        p_{e,1}(\yk) \geq  p_{\min} p_{e,1}(y_{1\ldots k+1}).
        \label{eq:bigger_of_two}
    \end{equation}
\end{lemma}
\begin{proof}
    We proceed via a straightforward computation:
    \begin{align}
    &p_{e,1}(\yk) \nn \\
    &= \frac{\p(\hat{\Theta}=0,\yk|\Theta=1)}{\p(\yk|\Theta=1)}\\
    &\geq \frac{\p(\hat{\Theta}=0,\yk,y_{k+1}|\Theta=1)}{\p(\yk|\Theta=1)}\\
    &= \frac{\p(\hat{\Theta}=0,\yk,y_{k+1}|\Theta=1)\cdot P(y_{k+1}|\vec{x}_1^{(k+1)})}{\p(\yk|\Theta=1)\cdot P(y_{k+1}|\vec{x}_1^{(k+1)})}\\
    &= \frac{\p(\hat{\Theta}=0,\yk,y_{k+1}|\Theta=1)\cdot P(y_{k+1}|\vec{x}_1^{(k+1)})}{\p(y_{1\ldots k+1}|\Theta=1)}\\
    &= P(y_{k+1}|\vec{x}_1^{(k+1)})p_{e,1}(y_{1\ldots k+1})\\
    &\geq p_{\min} p_{e,1}(y_{1\ldots k+1}).
    \end{align}
\end{proof}

\subsection{Chernoff Divergence and its Properties}

In this subsection, we let $Q$ be a generic binary-input discrete memoryless channel, and write $Q_x = Q(\cdot|x)$.  We also use $P$ (and sometimes $P'$) in a generic manner, either to mean a distribution on some alphabet $\calX$ or a conditional distribution on $\calY$ given $\calX$ (i.e., another discrete memoryless channel).  Accordingly, in this subsection, the notations $P,Q$ should not be taken to necessarily be related to the two channels in our problem setup.

For two probability distributions $P,P'$ over some finite set $\calX$, the Chernoff divergence with parameter $s$ ($0<s<1$) is given by
\begin{equation}
    \dc(P, P',s) =  -\log \sum_{x \in \calX}  P(x)^{1-s}P'(x)^{s}.
    \label{eq:dc}
\end{equation}
We also define $\dc(P,P',0)$ and $\dc(P,P',1)$ by continuity, e.g., $\dc(P, P', 0) = \lim_{s \rightarrow 0^+} \dc(P, P', s)$.  Recall from the discussion following Theorem \ref{thm:binary_optimal} that we only need to consider scenarios where $\dc$ is finite.

As an alternative for convenient notation, if $P$ is a memoryless channel then we define the Chernoff divergence associated with two channel inputs by 
\begin{gather}
    \dc(x, x', P,s) = \dc(P_x, P_{x'},s), \label{eq:dc_channel}
\end{gather}
where $P_x(\cdot) = P(\cdot|x)$.  For any positive integer $k$, we let $P^k$ denote the $k$-fold product of $P$, with probability mass function $P^k(\vec{y}|\vec{x}) = \prod_{i=1}^k P(y_i|x_i)$.  
For two sequences $\vec{x}, \vec{x}'$ of length $k$, we use the notation $\dc(\vec{x}, \vec{x}', P^k,s)$ similarly to \eqref{eq:dc_channel}, with the understanding that $\vec{x}, \vec{x}'$ are treated as inputs to $P^k$. 
We will use a well-known tensorization property of $\dc$, stated as follows.

\begin{lemma} 
{\em (e.g., see \cite[Lemma 8]{multibit})}
    For any DMC $P$ and any sequences $\vec{x} = (x_1,\dotsc,x_k)$ and $\vec{x}' = (x'_1,\dotsc,x'_k)$, we have
    \begin{equation}
        \dc(\vec{x}, \vec{x}', P^k, s) = \sum_{i=1}^k \dc(x_i, x'_i, P, s).
        \label{eq:iid1}
    \end{equation}
    \label{lem:iid}
\end{lemma}
We will also make use of certain derivatives of $\dc$; in particular, $\dc'$ and $\dc''$ represent the first and second derivatives with respect to the final argument.

The following straightforward lemma relates the Chernoff divergence with inputs $0,1$ to that with inputs $1,0$.

\begin{lemma} For any binary-input channel $Q$ and $s \in (0,1)$, we have
    \begin{equation}
    \dc'(1,0,Q,s) = -\dc'(0,1,Q,1-s).
    \label{eq:reverse1s}
    \end{equation}
    \label{lem:reverse1s}
\end{lemma}
\begin{proof}
Observe from \eqref{eq:dc} that
\begin{align}
    \dc'(1,0,Q,s) &= \frac{d}{dt} \dc(1,0,Q,t) |_{t=s}  \nn \\ &= \frac{d}{dt} \dc(0,1,Q,1-t) |_{t=s}.
\end{align}
Then, we can apply the chain rule to get $\frac{d}{dt} \dc(0,1,Q,1-t) |_{t=s} = -\dc'(0,1,Q,1-s)$.
\end{proof}
We will also use the following bound for $\dc''(Q_0,Q_1,s)$; while this result is known from \cite{berlekampI}, we provide a short proof because the intermediate calculations will also be useful. 

\begin{lemma}
{\em (Implicit in the proof of \cite[Thm.~5]{berlekampI})}
Let $p_{\min}$ be the smallest nonzero transition probability of $Q_0$ and $Q_1$. For all $s \in (0,1)$, we have
    \begin{equation}
        |\dc'(Q_0, Q_1, s) | \leq \log \frac{1}{p_{\min}}
        \label{eq:dc_d1}
    \end{equation}
    and
    \begin{equation}
        0 \leq -\dc''(Q_0, Q_1, s) \leq \left(\log \frac{1}{p_{\min}}\right)^2.
        \label{eq:dc_d2}
    \end{equation}
    \label{lem:d2c}
\end{lemma}
\begin{proof}
    Let $x \in \{0,1\}$ denote a generic input to $Q$.  
    We compute $\dc'$ and $\dc''$ as follows: 
    \begin{align}
        \dc'(Q_0, Q_1, s) &= -\frac{\frac{d}{ds} \sum_x Q_0(x)^{1-s}Q_1(x)^s}{\sum_x Q_0(x)^{1-s}Q_1(x)^s} \nn \\ &= -\frac{\sum_x Q_0(x)^{1-s}Q_1(x)^s \log \frac{Q_1(x)}{Q_0(x)}}{\sum_x Q_0(x)^{1-s}Q_1(x)^s}
    \end{align}
    \begin{align}
        &\dc''(Q_0, Q_1, s) \nn \\ &= -\frac{\sum_x Q_0(x)^{1-s}Q_1(x)^s \big(\log \frac{Q_1(x)}{Q_0(x)}\big)^2}{\sum_z Q_0(x)^{1-s}Q_1(x)^s} + \dc'(Q_0, Q_1, s)^2.
    \end{align}
    Let $Q_s$ be the distribution defined in \eqref{eq:def_qs}; then, we can write the above as
    \begin{equation}
        \dc'(Q_0, Q_1, s) = -\e_{x \sim Q_s} \log \frac{Q_1(x)}{Q_0(x)}
        \label{eq:dc_expected}
    \end{equation}
    \begin{equation}
        \dc''(Q_0, Q_1, s) = -{\rm Var}_{x \sim Q_s} \log \frac{Q_1(x)}{Q_0(x)}.
        \label{eq:dc_var}
    \end{equation}
    From \eqref{eq:dc_var}, we conclude the following:
    \begin{itemize}
        \item Since $-\dc''(Q_0, Q_1, s)$ is the variance of some random variable, $-\dc''(Q_0, Q_1, s)\geq 0$.
        \item All $x$ with nonzero mass under $Q_s$ must satisfy $Q_0(x)>0$ and $Q_1(x) > 0$. Therefore, for all $Q_s(x)>0$, 
        \begin{equation}
            \Big|\log \frac{Q_1(x)}{Q_0(x)}\Big| \leq \log \frac{1}{p_{\min}},
            \label{eq:ll_ratio}
        \end{equation}
        which implies
        \begin{equation}
            \hspace*{-3ex}|\dc'(Q_0, Q_1, s)| \leq \max_{Q_s(x)>0} \Big|\log \frac{Q_1(x)}{Q_0(x)}\Big| \leq \log \frac{1}{p_{\min}}
        \end{equation}
        and
        \begin{equation}
            \hspace*{-3ex}-\dc''(Q_0, Q_1, s) \leq \e_{x \sim Q_s} \left(\log \frac{Q_0(x)}{Q_1(x)}\right)^2 \leq \left( \log \frac{1}{p_{\min}}\right)^2. \label{eq:log_sq_bound}
        \end{equation}
    \end{itemize}
\end{proof}
\begin{lemma}
    Let $Q_0, Q_1$ be arbitrary probability distributions, and let $Q_s$ be defined via \eqref{eq:def_qs}.
    Then for all $s \in [0,1]$,
\begin{align}
    &\dq0  =  \dc(0,1,Q,s) - s\dc'(0,1,Q,s)\label{eq:ds0}\\ 
    &\dq1  =  \dc(0,1,Q,s) + (1-s)\dc'(0,1,Q,s)\label{eq:ds1}\\
    &\dqs0  = \dc(0,1,Q,1-s) - (1-s)\dc'(0,1,Q,1-s)\label{eq:dqs0}\\
    &\dqs1  = \dc(0,1,Q,1-s) + s\dc'(0,1,Q,1-s)\label{eq:dqs1}\\
    &\dc'(0,1,Q,s) = \dq1 - \dq0\label{eq:dq10}.
\end{align}
    \label{lem:dc_kl}
\end{lemma}
    \begin{proof}
    We first expand \eqref{eq:log_sq_bound} as
    \begin{equation}
    \dc'(Q_0,Q_1,s) = \sum_x Q_s(x) \log \frac{Q_0(x)}{Q_1(x)}.
    \label{eq:expected_s_log}
    \end{equation}
    %
    We also note from the definitions of $\dc$ and $Q_s$ that
    \begin{equation}
        Q_s(x) = {Q_0(x)^{1-s}Q_1(x)^s} \exp(\dc(Q_0,Q_1,s)).
    \end{equation}
    Combining these observations, we have
    \begin{align}
        &\dq0 \nn \\ &= \sum_x Q_s(x) \log \frac{Q_s(x)}{Q_0(x)}\\
        &= \sum_x Q_s(x) \log \frac{ {Q_0(x)^{1-s}Q_1(x)^s} \exp(\dc(Q_0,Q_1,s))}{Q_0(x)}\\
        &= \sum_x Q_s(x) \left(\log \frac{Q_1(x)^s}{Q_0(x)^s} + \dc(Q_0,Q_1,s)\right)\\
        &= \dc(Q_0, Q_1, s) -s \sum_x Q_s(x) \log \frac{Q_0(x)}{Q_1(x)}\\
        & = \dc(Q_0, Q_1, s) - s \dc'(Q_0,Q_1,s),
    \end{align}
    which completes the proof of \eqref{eq:ds0}. Substituting $s \rightarrow 1-s$ into \eqref{eq:ds0} then gives \eqref{eq:dqs0}.

    By switching $Q_0$ and $Q_1$ in \eqref{eq:ds0}, noting that $Q_s$ now becomes $Q_{1-s}$, we get
    \begin{align}
        \dqs1 &= \dc(1,0,Q,1-s) - (1-s)\dc'(1,0,Q,1-s) \nn \\ &\stackrel{\eqref{eq:reverse1s}}{=} \dc(0,1,Q,s) + (1-s)\dc'(0,1,Q,1-s)
    \end{align}
    which is equivalent to \eqref{eq:dqs1}. Substituting $s \rightarrow 1-s$ into \eqref{eq:dqs1} gives \eqref{eq:ds1}.  Finally, \eqref{eq:ds0}--\eqref{eq:ds1} immediately give \eqref{eq:dq10}.
\end{proof}

\section{Error Bounds for 1-Hop Communication}
\label{sec:1-hop}

In this section, we state some useful known results for the 1-hop setting, and adapt them to attain variations that will be more directly useful for our purposes.  We first consider the error probability of a binary hypothesis test as follows.

\begin{lemma}
{\em (\cite[Thm.~5]{berlekampI})} Let $\Pi$ and $\Pi'$ be two probability distributions defined over the same alphabet $\calZ$. Suppose that the space of all possible outcomes is partitioned into disjoint decoding regions $D_0$ and $D_1$, and define the following decoding error probabilities:
    \begin{equation}
        p_{e,1} = \p_{Z \sim \Pi}(Z \in D_1), p_{e,0} = \p_{Z \sim \Pi'}(Z \in D_0).
    \end{equation}
    Then, for any $s \in (0,1)$, at least one of the following must hold (it is possible for a different one to hold for different values of $s$): 
    \begin{align}
        -\log p_{e,0} &\leq \dc(\Pi, \Pi', s) - s \dc'(\Pi, \Pi', s) \nn \\ 
            &\qquad + s \sqrt{-2\dc''(\Pi, \Pi', s)} + \log 4 \label{eq:berlekamp_pe0}\\
        -\log p_{e,1} &\leq \dc(\Pi, \Pi', s) + (1-s) \dc'(\Pi, \Pi', s) \nn \\ & \qquad +  (1-s) \sqrt{-2\dc''(\Pi, \Pi', s)} + \log 4, \label{eq:berlekamp_pe1}
    \end{align}
    where $\dc'$ and $\dc''$ are first and second derivatives with respect to the final argument. 
    \label{lem:berlekamp}
\end{lemma}

We now let $Q$ be a binary-input DMC (again treated as generic for now) with inputs in $\calW = \{0,1\}$ and outputs in $\calZ$, and let $\w_0$ and $\w_1$ corresponding binary codewords of length $n$. Let $D_0, D_1$ be an arbitrary partition of the output space $\calZ^n$ into decoding regions, with decoding error probabilities defined by 
\begin{equation}
    p_{e,0} = \p_{\vec{z} \sim Q^n(\cdot|\w_0)} (\vec{z} \in D_1), ~~ p_{e,1} = \p_{\vec{z} \sim Q^n(\cdot|\w_1)} (\vec{z} \in D_0).
    \label{eq:decoding_error}
\end{equation}
Let $p_{\min}$ denote the smallest nonzero transition probability of $Q$. We obseve the following expressions relating to $\dc$ and its derivatives:
\begin{equation}
    \dc(\w_0, \w_1, Q^n, s) \stackrel{\text{Lem.}~\ref{lem:iid}}{=} \sum_{i=1}^n \dc(\w_0^{(i)}, \w_1^{(i)}, Q, s)
\end{equation}
\begin{equation}
    \dc'(\w_0, \w_1, Q^n, s) = \sum_{i=1}^n \dc'(\w_0^{(i)}, \w_1^{(i)}, Q, s)
\end{equation}
\begin{align}
     -\dc''(\w_0, \w_1, Q^n, s) &= \sum_{i=1}^n -\dc''(\w_0^{(i)}, \w_1^{(i)}, Q, s) \nn \\ &\stackrel{\text{Lem.}~\ref{lem:d2c}}{\leq} n  \left( \log \frac{1}{p_{\min}}\right)^2
\end{align}
\begin{equation}
     \sqrt{-\dc''(\w_0, \w_1, Q^n, s)} \stackrel{\text{Lem.}~\ref{lem:d2c}}{\leq} \sqrt n \log \frac{1}{p_{\min}}.
\end{equation}
Using these findings, we specialize Lemma \ref{lem:berlekamp} to the case of transmitting strings over a DMC as follows.

\begin{lemma}
{\em (\cite[Thm.~5]{berlekampI} for DMCs)}
Let $Q$ be a discrete memoryless channel, and let $D_0$,$D_1$ be a partition of $\calY^n$. Let $\w_0$ and $\w_1$ be arbitrary strings of length $n$, and let the decoding error probabilities $p_{e,0}, p_{e,1}$ be defined by \eqref{eq:decoding_error}. For any $s \in [0,1]$, at least one of the following holds:
    \begin{align}
        -\log p_{e,0} &\leq \dc(\w_0, \w_1, Q^n, s) - s \dc'(\w_0, \w_1, Q^n, s) \nn \\ &\qquad + \sqrt {2n} \log \frac{1}{p_{\min}} + \log 4
        \label{eq:berlekamp_dmc_pe0}\\
        -\log p_{e,1} &\leq \dc(\w_0, \w_1, Q^n, s) + (1-s) \dc'(\w_0, \w_1, Q^n, s)  \nn \\ &\qquad  + \sqrt{2n} \log \frac{1}{p_{\min}} + \log 4
        \label{eq:berlekamp_dmc_pe1}.
    \end{align}
    \label{lem:berlekamp_dmc}
\end{lemma}

An important corollary of Lemma \ref{lem:berlekamp_dmc} is the following, which is deduced via a direct combination of \cite[Corollary to Thm.~5]{berlekampI} and \eqref{eq:log_sq_bound}.  The idea is to optimize over $s$ in Lemma \ref{lem:berlekamp_dmc}; the derivative terms $\dc'$ therein vanish when the maximum is attained by $s \in (0,1)$, and the case of an endpoint $s \in \{0,1\}$ being optimal can also readily be handled separately.

\begin{cor}
    {\em (Implicit in \cite{berlekampI})} 
    Let $Q$ be a discrete memoryless channel, let $D_0$,$D_1$ be a partition of $\calZ^n$, and let $\w_0$ and $\w_1$ be arbitrary strings of length $n$. 
    Then, we have
    \begin{align}
        -\log (p_{e,0}+ p_{e,1}) &\leq \max_{0\leq s \leq 1}\dc(\w_0, \w_1, Q^n, s) \nn \\ &\qquad + \sqrt {2n} \log \frac{1}{p_{\min}} + \log 4.
    \end{align}
    \label{cor:2codewords}
\end{cor}

By combining this result with a tensorization argument (Lemma \ref{lem:iid}) and a fairly simple achievability analysis, we arrive at the following expression for the optimal 1-hop exponent.

\begin{cor}
    {\em \cite{berlekampI}}
    The 1-hop error exponent for any binary-input discrete memoryless channel $Q$ is given by
    \begin{equation}
        E(Q) = \max_{0\leq s \leq 1}\dc(0,1,Q,s) = \max_{0\leq s \leq 1}\dc(Q_0,Q_1,s).
        \label{eq:e1}
    \end{equation}
    \label{cor:e1}
\end{cor} 

Next, we provide a refined result that explicitly depends on the number of locations where $0$ and $1$ appear in $\w_0$ and $\w_1$ respectively. 

\begin{lemma}
Let $Q$ be a binary input channel with inputs $0$ and $1$, and let $\w_0$ and $\w_1$ be length-$n$ codewords. Let $n_{01}$ denote the number of positions $i$ such that $\w_0^{(i)}=0$ and $\w_1^{(i)}=1$, and let $n_{10}$ denote the number of positions $i$ such that $\w_0^{(i)}=1$ and $\w_1^{(i)}=0$.  Then, for any $s \in [0,1]$, at least one of the following holds:
\begin{align}
    -\log p_{e,0} &\leq n_{01}\dq0 + n_{10}\dqs1 \nn \\ &\qquad + \sqrt{2n} \log \frac{1}{p_{\min}}+\log 4
    \label{eq:pe0n} \\
    -\log p_{e,1} &\leq n_{01}\dq1 + n_{10}\dqs0  \nn \\ &\qquad + \sqrt{2n} \log \frac{1}{p_{\min}}+\log 4.
    \label{eq:pe1n}
\end{align}
\label{lem:nbound}
\end{lemma}
\begin{proof}
We first expand $\dc$ (via tensorization in Lemma \ref{lem:iid}) to get
\begin{align}
    &\dc(\w_0, \w_1, Q^n, s) \nn \\ &= \sum_{i=1}^n \dc(\w_0^{(i)}, \w_1^{(i)}, Q, s) \nn \\ &= n_{01}\dc(0,1,Q,s) + n_{10} \dc(0,1,Q,1-s)
\end{align}
and
\begin{align}
    &\dc'(\w_0, \w_1, Q^n, s) \nn \\ &~~= n_{01}\dc'(0,1,Q,s) + n_{10} \dc'(1,0,Q,1-s) \nn \\ &~~\stackrel{\eqref{eq:reverse1s}}{=} n_{01}\dc'(0,1,Q,s) - n_{10} \dc'(0,1,Q,1-s).
\end{align}
By using Lemma \ref{lem:dc_kl}, we observe that \eqref{eq:berlekamp_dmc_pe0} and \eqref{eq:pe0n} are equivalent, and \eqref{eq:berlekamp_dmc_pe1} and \eqref{eq:pe1n} are equivalent. By Lemma \ref{lem:berlekamp}, we conclude that \eqref{eq:pe0n} or \eqref{eq:pe1n} must be true.
\end{proof}

We now take Lemma \ref{lem:nbound} one step further, and state a variation regarding the error probability associated with the \emph{worse of two sequences} $\w_1$ and $\z_1'$.

\begin{lemma}
Let $Q$ be a binary-input channel, and let $\w_0$, $\w_1$ and $\z_1'$ be binary strings of length $n$.
Let $D_0$ and $D_1$ be any partition of the output space, and consider the following definitions for error probability:
\begin{gather}
    p_{e,0} = \p_{\vec{z} \sim Q^n(\cdot|\w_0)}(\vec{z} \in D_1), \nn \\ p_{e,1} = \max(\p_{\vec{z} \sim Q^n(\cdot|\w_1)}(\vec{z} \in D_0), \p_{\vec{z} \sim Q^n(\cdot|\w_1')}(\vec{z} \in D_0)).
\end{gather}
Recall the definitions of $n_0$ and $n_1$ from Definition \ref{dfn:n0_n1}, and suppose that $\ell$ is between\footnote{Both $n_1(\w_1)\leq \ell \leq n_1(\w'_1)$ and $n_1(\w_1)\geq \ell \geq n_1(\w'_1)$ are allowed, and $\ell$ need not be an integer.} $n_1(\z_1)$ and $n_1(\z'_1)$; then, for all $s \in (0,1)$, at least one of the following holds:
\begin{align}
    -\log p_{e,0} &\leq \ell \dq0 + (n-\ell) \dqs1 \nn \\ & \qquad + \sqrt{2n} \log \frac{1}{p_{\min}}+\log 4, \\
    -\log p_{e,1} &\leq \ell \dq1 + (n-\ell) \dqs0 \nn \\ & \qquad  + \sqrt{2n} \log \frac{1}{p_{\min}}+\log 4.
    \label{eq:transition_bound1}
\end{align}
\label{lem:transition_bound}
\end{lemma}
\begin{proof}
Without loss of generality, assume that $n_1(\z'_1) \geq \ell \geq n_1(\z_1)$ (if not, we simply exchange $\w_1$ and $\w'_1$). Then there are at least $n_1(\z'_1) - n_1(\z_1)$ positions $i$ with $\z'^{(i)}_1=1$ and $\z^{(i)}_1=0$. Among these $n_1(\z'_1) - n_1(\z_1)$ positions (if there are more, choose any subset of them of size $n_1(\z'_1) - n_1(\z_1)$), we consider the string $\z_0$, and conclude that one of the following must hold:
\begin{itemize}
    \item Case 1: There are at least $n_1(\z'_1)-\ell$ 1s (and these positions have 1 in both strings $\z'_1$ and $\z_0$)
    \item Case 2: There are at least $\ell - n_1(\z_1)$ 0s (and these positions have 0 in both strings $\z_1$ and $\z_0$)
\end{itemize}
We will now deal with these cases separately.

\underline{Case 1}: We apply Lemma \ref{lem:nbound} to the strings $\z_0$ and $\z'_1$. Since there are at least $n_1(\z'_1)-\ell$ positions in which both $\z'_1$ and $\z_0$ are 1, we have
\begin{equation}
    n_{01} \leq n_1(\z'_1) - (n_1(\z'_1)-\ell) = \ell
\end{equation}
and
\begin{equation}
    n_{10} \leq  n_0(\z'_1) = n-n_1(\z'_1) \leq n-\ell
\end{equation}
and thus the conclusion follows.

\underline{Case 2}: We apply Lemma \ref{lem:nbound} to the strings $\z_0$ and $\z_1$. Since there are at least $\ell - n_1(\z_1)$ positions in which both $\z_1$ and $\z_0$ are 0, we have
\begin{align}
    n_{10} &\leq n_0(\z_1) - (\ell - n_1(\z_1)) \nn \\ &= (n-n_1(\z_1)) - (\ell - n_1(\z_1)) \leq n-\ell
\end{align}
and
\begin{equation}
    n_{01} \leq n_1(\z_1) \leq \ell
\end{equation}
and the conclusion again follows.
\end{proof}

\section{Intuitive Overview of the Proof} \label{sec:overview}

Since the proof of Theorem \ref{thm:binary_optimal_exact} is quite technical, we provide an informal intuitive overview in this section before turning to the full details.  In this section (only), we omit the explicit remainder terms in the exponent, and instead simply write them as $o(n)$.

Recall that the exponent in \eqref{eq:2hop_optimal} contains a minimum over a $P$ part and a $Q$ part, each of which contains a maximum associated with $s$ and $1-s$.  By initially handling some easier cases, we will find that the hardest case to handle has the $P$ part and $Q$ part being equal, which in turn means that the minimum of the two is the same as the maximum.  Accordingly, recalling that $E^*$ is the target exponent in \eqref{eq:2hop_optimal}, we suppose that $s$ and $E^*$ satisfy
\begin{multline}
    \max\{\dc(0,1,P,s), \dc(0,1,P,1-s), \\ \dc(0,1,Q,s), \dc(0,1,Q,1-s)\} \leq E^* \label{eq:four_terms_E}
\end{multline}
and seek to prove that $E(P,Q)\leq E^*$.
For any non-negative integer $k\leq n$, we can calculate the total error probability by conditioning on the first $k$ symbols received by the relay:
\begin{align}
    p_{e,0}+p_{e,1} = \sum_{\yk\in \calY^k} \Big( &P(\yk|\vec{x}_0) p_{e,0}(\yk) \nn \\ + &P(\yk|\vec{x}_1) p_{e,1}(\yk) \Big),
    \label{eq:length_k_mass}
\end{align}
where we recall that $P(\yk|\vec{x}_0)$ is a shorthand for $P^k(\yk|\vec{x}_0)$ with $P^k$ being the $k$-fold product distribution, and $p_{e,0}(\yk)$ and $p_{e,1}(\yk)$ denote conditional error probabilities given $\yk$.  Observe that for any two strings $\vec{w}$ and $\vec{w}'$ of length $n$ such that the first $k$ symbols are identical, we have
\begin{equation}
    \dc(\vec{w}, \vec{w}', Q^n, s) = \sum_{i=1}^{n} \dc(\vec{w}^{(i)}, \vec{w}'^{(i)}, Q, s) \leq (n-k) E^*
\end{equation}
because the first $k$ terms are zero and the rest are bounded via \eqref{eq:four_terms_E}.

Recall that $\w_0(\yk)$ and $\w_1(\yk)$ are defined in \eqref{eq:w1_def}--\eqref{eq:w0pe0}, and essentially amount to fixing the first $k$ symbols sent by the relay (since they are determined via $\yk$) and considering the rest to be the best possible in terms of conditional error probability (given $\Theta=0$ or $\Theta=1$).  Due to these two sequences being identical in the first $k$ symbols, Lemma \ref{lem:berlekamp} tells us that 
either
\begin{align}
    &-\log p_{e,0}(\yk) \nn \\ & \leq  \dc(\w_0(\yk), \w_1(\yk),Q^n,s) \nn \\ 
        &\qquad - s\dc'(\w_0(\yk), \w_1(\yk),Q^n,s) + o(n) \\
    & \leq (n-k)E^* - s\dc'(\w_0(\yk), \w_1(\yk),Q^n,s) + o(n)
    \label{eq:wishful0}
\end{align}
or
\begin{align}
    &-\log p_{e,1}(\yk) \nn \\ & \leq \dc(\w_0(\yk), \w_1(\yk),Q^n,s) \nn \\ 
        &\qquad 
     + (1-s)\dc'(\w_0(\yk), \w_1(\yk),Q^n,s) + o(n) \\
    & \leq (n-k)E^* + (1-s)\dc'(\w_0(\yk), \w_1(\yk),Q^n,s)  \nn \\ 
        &\hspace*{6.5cm} + o(n).
    \label{eq:wishful1}
\end{align}
We note that for the most part, instead of considering $\w_0(\cdot)$ and $\w_1(\cdot)$, our analysis will consider $\w_1^c(\cdot)$ and $\w_1(\cdot)$, with $\w_1^c(\cdot)$ being the bit-wise complement of $\w_1(\cdot)$.  This will be justified by the fact that such a substitution increases $\dc$, and the corresponding changes in $\dc'$ can also be justified using Lemma \ref{lem:dc_kl}.  To avoid over-complicating this section, we always use $\w_0(\cdot)$ and $\w_1(\cdot)$ in the following outline (at the expense of parts of the outline matching the actual proof steps less precisely).

Suppose that by some ``wishful thinking'' (to be addressed in more detail below), we have
\begin{equation}
    \dc'(\w_0(\yk), \w_1(\yk),Q^n,s) = -\log \frac{P(\yk|\vec{x}_0)}{P(\yk|\vec{x}_1)} + o(n).
    \label{eq:wishful}
\end{equation}
Then, under \eqref{eq:wishful0}, we would have
\begin{align}
    & p_{e,0}(\yk) P(\yk|\vec{x}_0) \nn \\ &\geq  \exp(-(n-k)E^*+o(n)) \left(\frac{P(\yk|\vec{x}_0)}{P(\yk|\vec{x}_1)}\right)^{-s} P(\yk|\vec{x}_0)\\
    & \geq \exp(-(n-k)E^*+o(n))  P(\yk|\vec{x}_0)^{1-s} P(\yk|\vec{x}_1)^{s},
\end{align}
and similarly, under \eqref{eq:wishful1}, we would have
\begin{align}
    & p_{e,1}(\yk) P(\yk|\vec{x}_1) \nn \\ & \geq \exp(-(n-k)E^*+o(n)) \left(\frac{P(\yk|\vec{x}_0)}{P(\yk|\vec{x}_1)}\right)^{1-s}  P(\yk|\vec{x}_1)\\
    & \geq \exp(-(n-k)E^*+o(n)) P(\yk|\vec{x}_0)^{1-s} P(\yk|\vec{x}_1)^{s}.
\end{align}
Moreover, a standard ``single-letterization'' argument gives
\begin{align}
    &\sum_{(\yk) \in \calY^k} P(\yk|\vec{x}_0)^{1-s} P(\yk|\vec{x}_1)^{s} \nn \\
 & = \sum_{y_1 \in \calY} \sum_{y_2 \in \calY} \ldots \sum_{y_k \in \calY} \prod_{i=1}^k P(y_i|\vec{x}_0^{(i)})^{1-s} P(y_i|\vec{x}_1^{(i)})^{s}\\
 & = \exp\left(-\sum_{i=1}^k \dc(\vec{x}_0^{(i)}, \vec{x}_1^{(i)}, P,s)\right) \ge e^{-kE^*},
\end{align}
with the inequality using \eqref{eq:four_terms_E}. 
Combining the above with \eqref{eq:wishful0} and \eqref{eq:wishful1} would then give an overall error bound of
\begin{align}
    &p_{e,0} + p_{e,1} \nn \\ 
   &\ge \exp(-(n-k)E^* + o(n)) \cdot \exp\left(-\sum_{i=1}^k \dc(\vec{x}_0^{(i)}, \vec{x}_1^{(i)}, P, s)\right) \nn \\ &\geq \exp(-nE^* + o(n)).
\end{align}

The preceding argument is based on the hypothetical event of \eqref{eq:wishful} being true.  Unfortunately, however, we cannot guarantee that this is the case for arbitrary sequences $\yk$.  Instead, we seek to argue that the quantity
\begin{equation}
    \dc'(\w_0(\yk), \w_1(\yk),Q^n,s) + \log \frac{P(\yk|\vec{x}_0)}{P(\yk|\vec{x}_1)}
    \label{eq:small_abs}
\end{equation}
is ``sufficiently small'' in absolute value for a ``sufficiently large set'' of $\yk$ vectors.

Towards achieving this goal, we note a simple but important observation: If $\yn$ is such that the quantities $\dc'(\w_0(\emptyset), \w_1(\emptyset),Q^n,s)$ and $\log \frac{P(\yn|\vec{x}_0)}{P(\yn|\vec{x}_1)}$ are of opposite sign, we can find some $k$ and $k+1$ at which \eqref{eq:small_abs} changes sign. 
It is easy to check that $\log \frac{P(\yk|\vec{x}_0)}{P(\yk|\vec{x}_1)}$ changes by at most $\log \frac{1}{p_{\min}}$ when we change $k$ to $k+1$.  We would ideally have a similar ``smoothness'' (i.e., only small/gradual changes) for $\dc'(\w_0(\yk), \w_1(\yk),Q^n,s)$, as this would imply that \eqref{eq:small_abs} is small at this $k$ value.  Unfortunately, the argument is not this simple, because in principle $\dc'(\w_0(\yk), \w_1(\yk),Q^n,s)$ can change by a large amount from $k$ to $k+1$, so we need a way to essentially ``smooth'' this out.  We also need to account for the fact that different sequences $\yn$ can have different values of $k$ at which \eqref{eq:small_abs} switches sign.  

We elaborate on these main difficulties (with forward-references to where they are addressed) as follows:
\begin{itemize}
\item To address the fact that $\dc'(\w_0(\yk), \w_1(\yk),Q^n,s)$  can jump by a large amount as we move from $k$ to $k+1$, we set up bounds on $-\log p_{e,0}(\yk)$ and $-\log p_{e,1}(\yk)$ using $\ell_k$, the number of 1s in the last $n-k$ symbols of $\w_1$, instead of $\dc'(\w_0(\yk), \w_1(\yk),Q^n,s)$ itself (Lemma \ref{lem:l_count}).  There is still no guarantee that $|\ell_{k+1}-\ell_k|$ will be small, but Lemma \ref{lem:transition_bound} allows us to make conclusions about values of $\ell$ in between $\ell_k$ and $\ell_{k+1}$.
\item As mentioned above, the transition point $k$ may differ depending on the choice of $\yn$, making the error probability decomposition in  \eqref{eq:length_k_mass} difficult to work with for bounding $p_{e,0}(\yk)$ and $p_{e,1}(\yk)$. 
To address this, we will introduce (in Lemma \ref{lem:prefix_free}) a set $A$ of variable-length sequences on $\calY$ such that any $\yn$ has exactly one prefix in $A$, and use the following decomposition instead of \eqref{eq:length_k_mass}:
\begin{align}
    p_{e,0}+p_{e,1} = \sum_{\yk\in A} \Big( &P(\yk|\vec{x}_0) p_{e,0}(\yk) \nn \\ + &P(\yk|\vec{x}_1) p_{e,1}(\yk) \Big).
\end{align}
\item Recall that in the preceding outline, we imposed the assumption that $\dc'(\w_0(\emptyset), \w_1(\emptyset),Q^n,s)$ and $\log \frac{P(\yn|\vec{x}_0)}{P(\yn|\vec{x}_1)}$ are of opposite sign.  When their signs are the same, it appears to be difficult to guarantee that \eqref{eq:small_abs} is small.  Lemma \ref{lem:finish} addresses this issue by showing that a closely related sign condition has a ''sufficiently small'' probability of occurring under a suitably-chosen distribution.\footnote{Lemma \ref{lem:finish} actually concerns the signs of $\dc'(\vec{x}_0, \vec{x}_1, P^n, t)$ and $\log \frac{P(\yn|\vec{x}_0)}{P(\yn|\vec{x}_1)}$, but by that point of the analysis, (i) we will be considering $\dc'(\w_1^c(\emptyset), \w_1(\emptyset), Q^n, t)$ instead of $\dc'(\w_0(\emptyset), \w_1(\emptyset), Q^n, t)$ in accordance with the discussion following \eqref{eq:wishful1}, and (ii) we will have reduced to a case where $\dc'(\vec{x}_0, \vec{x}_1, P^n, t)$ and $\dc'(\w_1^c(\emptyset), \w_1(\emptyset), Q^n, t)$ are of strictly opposite sign; see Assumption \ref{assump:opp}.}
\end{itemize}
All of these aspects of the analysis turn out to be quite technical and complicated.

\section{Simplifying Conditions Towards Proving Theorem \ref{thm:binary_optimal_exact}} \label{sec:reduction}

In this section, we identify some simplifying conditions that can be assumed when proving Theorem \ref{thm:binary_optimal_exact}.  We define
\begin{align}
    E_s = \min\Big(
    &\max\big(\dc(P_0, P_1, s), \dc(P_0, P_1, 1-s)\big), \nn \\
    &\max\big(\dc(Q_0, Q_1, s), \dc(Q_0, Q_1, 1-s)\big)
    \Big)
    \label{eq:e_s}
\end{align}
\begin{equation}
    E^* = \max_{0\leq s \leq 1} E_s,
    \label{eq:e_def}
\end{equation}
and let $s^{*}$ be the corresponding argmax (which exists since $E_s$ is continuous in $s$). Recall that the goal of Theorem \ref{thm:binary_optimal} is to show that $E(P,Q) \leq E^*$ (and Theorem \ref{thm:binary_optimal_exact} gives a more precise statement).

Since $E_s = E_{1-s}$, we have $E^* = \max_{0\leq s \leq \frac12} E_s$, so we may assume that $s^* \leq \frac12$.  Moreover, we may assume that $\dc(P_0, P_1, s^{*}) \ge \dc(P_0, P_1, 1-s^{*})$; this is because if not, we can simply reverse the inputs of $P$. Similarly, we may assume that $\dc(Q_0, Q_1, s^{*}) \geq \dc(Q_0, Q_1, 1-s^{*})$. Under these assumptions, we have the simplified expression $E^* = \min(\dc(P_0, P_1,s^{*}), \dc(Q_0, Q_1, s^{*}))$.

In summary, we can make the following assumption.

\begin{assump} \label{assump:s_star}
We assume that $s^{*}\leq \frac12$, $\dc(P_0, P_1, s^{*}) \geq \dc(P_0, P_1, 1-s^{*})$, $\dc(Q_0, Q_1, s^{*}) \geq \dc(Q_0, Q_1, 1-s^{*})$, and $E^* = \min(\dc(P_0, P_1,s^{*}), \dc(Q_0, Q_1, s^{*}))$.
\end{assump}

We will investigate the behavior of $E_s $ near $s^{*}$ and derive some useful properties of $\dc'(\vec{x}_0, \vec{x}_1, P^n, s^{*})$ and $\dc'(\w_0(\emptyset), \w_1(\emptyset), Q^n, s^{*})$.  Our main goal in this section is to establish the following.

\begin{lemma}
    For $E^*$ defined in \eqref{eq:e_def}, if Assumption \ref{assump:s_star} holds and
    \begin{equation}
        -\log (p_{e,0} + p_{e,1}) > nE^* + \sqrt {2n} \log \frac{1}{p_{\min}} + \log 4,
        \label{eq:perr_opp}
    \end{equation}
    then there exists $t \in \{s^*, 1-s^*\}$ (with $s^* \in \big(0,\frac{1}{2}\big]$) such that $\dc'(\vec{x}_0, \vec{x}_1, P^n, t)$ and $\dc'(\w_1^c(\emptyset), \w_1(\emptyset), Q^n, t)$ (with $\w_1^c(\emptyset)$ being the bit-wise complement string of $\w_1(\emptyset)$) are of strictly opposite sign, and
    \begin{multline}
    \max(\dc(P_0, P_1,t), \dc(P_0, P_1,1-t)) \\ = \max(\dc(Q_0, Q_1,t), \dc(Q_0, Q_1,1-t)) = E^*.
    \label{eq:maxpq}
    \end{multline}
    \label{lem:opp_assump}
\end{lemma}

Note that if \eqref{eq:perr_opp} fails then Theorem \ref{thm:binary_optimal_exact} follows immediately, so Lemma \ref{lem:opp_assump} is giving us additional conditions that we can assume beyond Assumption \ref{assump:s_star}.  In particular, the condition of strictly opposite sign will be useful in view of the major role of signs outlined in Section \ref{sec:overview}, and the condition \eqref{eq:maxpq} indicates that we have reduced to the scenario where $P$ and $Q$ contribute the same value in the exponent, rather than one being dominant over the other.

%
%
The following definition will be central to the proof of Lemma \ref{lem:opp_assump}, providing a useful means to distinguish various cases depending on how $\dc(P_0,P_1,s)$ and $\dc(Q_0,Q_1,s)$ behave at the optimal $s$ value (increasing, decreasing, or a local maximum).

\begin{dfn} \label{def:types}
    Let $g:[0,1] \rightarrow \mathbb{R}_{\geq 0}$ be a concave function with a unique global maximum, and fix $s^* \in [0,\frac12]$. We say that a function is \emph{skewed} if $\argmax_s g(s) < s^*$ or $\argmax_s g(s) > 1-s^*$,  \emph{balanced} if $s^* < \argmax_s g(s) < 1-s^*$, and \emph{neutral} if $\argmax_s g(s) \in \{s^*, 1-s^*\}$ (see Figure \ref{fig:skewed}).

    We will also use the term \emph{weakly balanced} to mean balanced or neutral, and \emph{weakly skewed} to mean skewed or neutral.
    Given two functions, we say that they are of \emph{strictly opposite type} if one is skewed and the other is balanced, and \emph{weakly opposite type} if one is weakly skewed and the other is weakly balanced.
\end{dfn}

Before proceeding, we give an overview of how Lemma \ref{lem:opp_assump} will be proved and the relevant intermediate results towards doing so:
\begin{itemize}
    \item If $E^* = E(P)$ or $E^* = E(Q)$ (i.e., the target exponent $E^*$ matches a 1-hop exponent) then a simple argument based on Corollary \ref{cor:2codewords} reveals that \eqref{eq:perr_opp} cannot hold.  Thus, we may assume that $E^* \ne E(P)$ and $E^* \ne E(Q)$.
    \item When $E^* \ne E(P)$ and $E^* \ne E(Q)$, we will show in Lemma \ref{lem:dcp} that the only remaining possibility is that $E^* = \dc(P_0, P_1, s^*) = \dc(Q_0, Q_1, s^*)$, $s^*\neq 0$, and the quantities $\dc(P_0, P_1, s^*)$ and $\dc(Q_0, Q_1, s^*)$ are of strictly opposite type in the sense of Definition \ref{def:types}.  The former of these statements will immediately give \eqref{eq:maxpq}.
    \item Lemmas \ref{lem:opp_p} and \ref{lem:opp_q} will provide additional conditions under which \eqref{eq:perr_opp} cannot hold, which will imply that the preceding information on $\dc(P_0, P_1, s^*)$ and $\dc(Q_0, Q_1, s^*)$ translates to analogous information on $\dc(\vec{x}_0, \vec{x}_1, P^n, s)$ and $\dc(\w_1^c(\emptyset), \w_1(\emptyset), Q^n, s)$, namely, the two have the same non-neutral type.  Combining this finding with the definitions of skewed and balanced, we will deduce the desired result on $\dc'(\vec{x}_0, \vec{x}_1, P^n, t)$ and $\dc'(\w_1^c(\emptyset), \w_1(\emptyset), Q^n, t)$ having opposite signs.
\end{itemize}
We now proceed to make the above outline concrete.

We will use Definition \ref{def:types} with $\dc(P_0,P_1,s)$ or $\dc(Q_0,Q_1,s)$ playing the role of $g$ (treated as functions of $s$ with the other arguments being fixed), so we proceed to discuss the assumption of concavity and a unique maximum under these choices.  Concavity immediately follows from \eqref{eq:dc_var}, and this equation additionally implies that $\dc''$ is either negative for all $s \in (0,1)$ or zero for all $s \in (0,1)$ (since $Q_s$ has the same support for all such $s$).  In the former case, we have strict concavity in $s$, which trivially implies a unique maximum.  On the other hand, if $\dc''$ is zero, then $\dc$ is a straight line, and there are two cases to consider:
\begin{itemize}
    \item If $\dc$ has a non-zero slope (e.g., this occurs for the Z-channel), then the maximum is again unique, and occurs either at $s=0$ or $s=1$.
    \item If $\dc$ has slope zero (e.g., this occurs for the binary erasure channel), then the maximum is non-unique.  However, in this case the max-min order in \eqref{eq:2hop_optimal} is irrelevant anyway, so the trivial converse from \cite{teachlearn} (i.e., 2-hop is at least as hard as 1-hop) can be applied directly.
\end{itemize}
In view of the second dot point, we subsequently assume that both $\dc(P_0,P_1,s)$ and $\dc(Q_0,Q_1,s)$ are non-constant with respect to $s$, so they both satisfy the conditions of Definition \ref{def:types}.

\begin{figure*}
\begin{center}
\begin{tikzpicture}[scale=0.7]
\begin{axis}[
    axis lines = left,
    xlabel = {$s$},
    ylabel = {}, 
    yticklabels={,,} 
]
\addplot [
    domain=0:1, 
    samples=100, 
    color=red,
]
{-log2(0.999^x*0.5^(1-x)+0.001^x*0.5^(1-x))};

\addplot[dashed] coordinates{(0.4,0) (0.4,0.5)};
\addplot[dashed] coordinates{(0.6,0) (0.6,0.36)};
\addplot [only marks] table {
0.317 0.53
};
\end{axis}
\node at (3,-2) {skewed};
\end{tikzpicture}
\begin{tikzpicture}[scale=0.7]
\begin{axis}[
    axis lines = left,
    xlabel = {$s$},
    ylabel = {},
    yticklabels={,,} 
]
\addplot [
    domain=0:1, 
    samples=100, 
    color=red,
]
{-log2(0.8^x*0.2^(1-x)+0.2^x*0.8^(1-x))};
\addplot [only marks] table {
0.5 0.32
};
\addplot[dashed] coordinates{(0.4,0) (0.4,0.3)};
\addplot[dashed] coordinates{(0.6,0) (0.6,0.3)};
\end{axis}
\node at (3,-2) {balanced};
\end{tikzpicture}
\begin{tikzpicture}[scale=0.7]
\begin{axis}[
    axis lines = left,
    xlabel = {$s$},
    ylabel = {},
    yticklabels={,,} 
]   
\addplot [
    domain=0:1, 
    samples=100, 
    color=red,
]
{-log2(0.981^x*0.5^(1-x)+0.019^x*0.5^(1-x))};
\addplot [only marks] table {
0.4 0.34
};
\addplot[dashed] coordinates{(0.4,0) (0.4,0.34)};
\addplot[dashed] coordinates{(0.6,0) (0.6,0.28)};
\end{axis}
\node at (3,-2) {neutral};
\end{tikzpicture}
\end{center}
    \caption{Examples of skewed, balanced, and neutral. In this case, we are using $s^*=0.4$ (and therefore $1-s^*=0.6$). The solid circle represents the global maximum. }
    \label{fig:skewed}
    
\end{figure*}
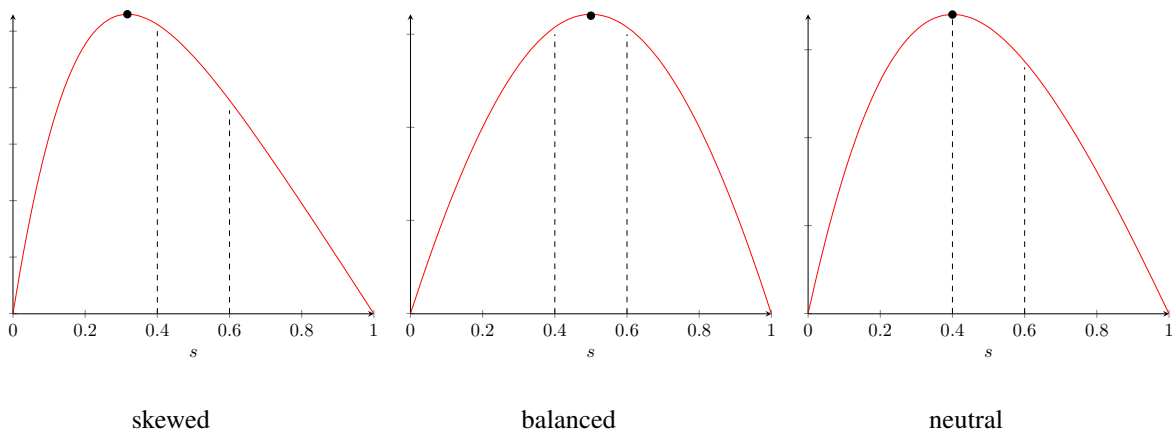
From the definitions of skewed and balanced (and the assumption that $g$ has a unique maximizer), we immediately obtain the following.  

\begin{fact}
    If $g$ is skewed, then $g'(s^*)$ and $g'(1-s^*)$ have the same sign and $g$ is monotone on $[s^*, 1-s^*]$. If $g$ is balanced, then it is increasing on $[0,s^*]$ and decreasing on $[1-s^*,1]$, with $g'(s^*)>0$ and $g'(1-s^*)<0$.
    \label{fact:monotone}
\end{fact}

The following lemma will be useful for inferring conditions under which $\dc(P_0, P_1, s^*)$ and $\dc(Q_0, Q_1, s^*)$ are of strictly opposite type (in the sense of Definition \ref{def:types}).  Recall from the above outline that this will later be translated to similar kinds of conditions on $\dc(\vec{x}_0, \vec{x}_1, P^n, s)$ and $\dc(\w_1^c(\emptyset), \w_1(\emptyset), Q^n, s)$.

\begin{lemma}
    For any pair $(P,Q)$ of binary-input DMCs, at least one of the following holds:
    \begin{itemize}
        \item $E^* = E(P)$;
        \item $E^* = E(Q)$;
        \item $E^* = \dc(P_0, P_1, s^*) = \dc(Q_0, Q_1, s^*)$, $s^*\neq 0$, and the quantities $\dc(P_0, P_1, s^*)$ and $\dc(Q_0, Q_1, s^*)$ are of strictly opposite type.
    \end{itemize}
    Here, $E(P)$ and $E(Q)$ denote the respective optimal 1-hop error exponents (see Corollary \ref{cor:e1}).
    \label{lem:dcp}
\end{lemma}
\begin{proof}
We split into cases as follows:
\begin{itemize}
    \item Case 1: $\dc(P_0, P_1, s^*) > \dc(Q_0, Q_1,s^*)$. 
    \begin{itemize}
        \item Case 1a: $\dc(Q_0,Q_1,s)$ attains a global maximum at $s=s^*$.  Then, we have
    \begin{align}
    E^* &= E_{s^*} \stackrel{\text{Assump.}~\ref{assump:s_star}}{=} \dc(Q_0, Q_1, s^*) \nn \\ &= \max_{0\leq s \leq 1} \dc(Q_0, Q_1, s) \stackrel{\eqref{eq:e1}}{=} E(Q).
    \label{eq:epq_case2}
    \end{align}
    \item Case 1b: $\dc(Q_0,Q_1,s)$ does not attain a global maximum at $s=s^*$. Since $\dc(Q_0, Q_1, \cdot)$ is concave, $\dc(Q_0, Q_1, \cdot)$ cannot have a local maximum at $s^*$.\footnote{If $s^*=0$ and $\dc(Q_0, Q_1, t) < \dc(Q_0, Q_1,0)$ for sufficiently small $t$, we also consider this to be a ``local maximum''.  Such a scenario cannot fall under Case 1b, because it contradicts the assumption that $\dc(Q_0,Q_1,s)$ does not attain a global maximum at $s=s^*$.}  Then, by moving slightly in an increasing direction from $s^*$, we can find some $s^{**}$ such that $\dc(Q_0, Q_1, s^{**}) > \dc(Q_0, Q_1,s^*)$, while maintaining $\dc(P_0, P_1, s^{**})>\dc(Q_0,Q_1,s^*)$ (due to continuity), so that $E_{s^{**}} > E_{s^*}$.  This contradicts the assumption that ${s^*}$ is the maximizer of $E_s$.
    \end{itemize}
    \item Case 2: $\dc(P_0, P_1, s^*) < \dc(Q_0, Q_1,s^*)$. Similarly to Case 1, we can conclude that $E^* = E(P)$.
    \item Case 3: $\dc(P_0, P_1, s^*) = \dc(Q_0, Q_1,s^*)$. 
    \begin{itemize}
        \item Case 3a: If $\dc(P_0, P_1,s)$ or $\dc(Q_0, Q_1,s)$ attains a global maximum at $s=s^*$, then we immediately have $E^* = E(P)$ or $E^* =E(Q)$ by \eqref{eq:e1}.
        \item Case 3b: If $s^*=0$ and neither $\dc(P_0, P_1,s)$ nor $\dc(Q_0, Q_1,s)$ attains a global maximum at $s=0$, then for all $s^{**}$ sufficiently close to $0$, $\dc(Q_0, Q_1, s^{**}) > \dc(Q_0, Q_1,0)$ and $\dc(P_0, P_1, s^{**})>\dc(Q_0,Q_1,0)$, so that $E_{s^{**}} > E_{0}$, contradicting the assumption that $s^*=0$. Therefore, if $s^*=0$, then $\dc(P_0, P_1,s)$ or $\dc(Q_0, Q_1,s)$ attains a global maximum at $s=0$, so that $E^* = E(P)$ or $E^* =E(Q)$.
        \item Case 3c: Suppose that neither of Cases 3a and 3b apply, and thus $s^* \neq 0$.  If $\dc(P_0, P_1, s)$ and $\dc(Q_0, Q_1, s)$ are both balanced or both skewed, then by Fact \ref{fact:monotone}, there exists some $s^{**}$ (obtained by shifting $s^*$ very slightly in the increasing direction) such that $\dc(P_0,P_1,s^{**})>\dc(P_0,P_1,s^*)$ and $\dc(Q_0,Q_1,s^{**})>\dc(Q_0,Q_1,s^*)$, implying 
        \begin{align}
            E_{s^{**}} \geq \min&(\dc(P_0, P_1, s^{**}, \dc(Q_0, Q_1, s^{**})) \nn \\ > \min&(\dc(P_0, P_1, s^{*})
            , \dc(Q_0, Q_1, s^{*}))  \nn \\ \stackrel{\text{Assump.~\ref{assump:s_star}}}=&E_{s^*}.
        \end{align}
     If $\dc(P_0, P_1,s)$ is neutral, then $\dc(P_0, P_1, s^*)$ is a global maximum by definition, and $E^* = E(P)$. Similarly, if $\dc(Q_0, Q_1, s)$ is neutral, then $E^* = E(Q)$.
    Therefore, we conclude that $\dc(P_0, P_1, s^*)$ and $\dc(Q_0, Q_1, s^*)$ must be of strictly opposite type.
    \end{itemize}
\end{itemize}
\end{proof}

Recall that Lemma \ref{lem:opp_assump} is based on the condition \eqref{eq:perr_opp}.  In the following lemma, we establish a useful scenario under which such a condition provably does \emph{not} hold. 
 This lemma will allow us to translate ``type'' information on $\dc(P_0,P_1, s)$ to that of $\dc(\vec{x}_0, \vec{x}_1, P^n, s)$, and Lemma \ref{lem:opp_q} will similarly relate $\dc(Q_0,Q_1,s)$ to $\dc(\w_1^c(\emptyset), \w_1(\emptyset), Q^n, s)$.

\begin{lemma}
    If $\dc(\vec{x}_0, \vec{x}_1, P^n,s)$ and $\dc(P_0,P_1,s)$ are of weakly opposite type, then
    \begin{equation}
        -\log (p_{e,0} + p_{e,1}) \leq n \dc(P_0, P_1, s^*) + \sqrt {2n} \log \frac{1}{p_{\min}} + \log 4.
        \label{eq:opp_p}
    \end{equation}
    \label{lem:opp_p}
\end{lemma}
\begin{proof}
    Using data processing inequalities and Corollary \ref{cor:2codewords}, we have
    \begin{align}
        &-\log (p_{e,0} + p_{e,1}) \leq -\log (\hat{p}_{e,0} + \hat{p}_{e,1}) \nn \\ &\leq \max_{0\leq s \leq 1}\dc(\vec{x}_0, \vec{x}_1, P^n, s) + \sqrt {2n} \log \frac{1}{p_{\min}} + \log 4,
        \label{eq:dpi}
    \end{align}
    where $\hat{p}_{e,0}$ and $\hat{p}_{e,1}$ refer to the error probabilities associated with the relay's (hypothetical) best estimate of $\Theta$, which is an clearly easier task than that of the decoder.  It therefore suffices to show that
    \begin{equation}
        \max_{0\leq s \leq 1}\dc(\vec{x}_0, \vec{x}_1, P^n, s) \leq n\dc(P_0, P_1, s^*).
    \end{equation}
    We proceed with three cases:
\begin{itemize}
    \item Case 1: Suppose that $\dc(P_0, P_1, s)$ is neutral. This means that $\dc(P_0, P_1, s)$ attains a global maximum at $s^*$, and hence
    \begin{align}
        \max_{0\leq s \leq 1}\dc(\vec{x}_0, \vec{x}_1, P^n, s) &\stackrel{\text{Lem.}~\ref{lem:iid}}{\leq} n \max_{0\leq s \leq 1} \dc(0, 1, P, s) \nn \\ &= n \dc(P_0, P_1, s^*).
    \end{align}
    \item Case 2: If $\dc(\vec{x}_0, \vec{x}_1, P^n, s)$ is weakly balanced and $\dc(P_0, P_1, s)$ is skewed, then
    \begin{align}
        &\max_{0\leq s \leq 1}\dc(\vec{x}_0, \vec{x}_1, P^n, s) \nn \\
        &~~= \max_{s^*\leq s \leq 1-s^*}\dc(\vec{x}_0, \vec{x}_1, P^n, s) \\ 
        &\stackrel{\text{Lem.}~\ref{lem:iid}}{\leq} \sum_{i=1}^n \max_{s^*\leq s \leq 1-s^*} \dc(\vec{x}_0^{(i)}, \vec{x}_1^{(i)}, P, s).
    \end{align}
On the other hand, since $\dc(P_0, P_1, s)$ is skewed, by Fact \ref{fact:monotone}, we have
    \begin{align}
        &\max_{s^*\leq s \leq 1-s^*} \dc(\vec{x}_0^{(i)}, \vec{x}_1^{(i)}, P, s) \nn \\
        &~~~\,= \max (\dc(0,1,P,s^*), \dc(0,1,P,1-s^*))  \nn \\ 
        &\stackrel{\text{Assump.~\ref{assump:s_star}}}{=} \dc(P_0, P_1, s^*).
    \end{align} 
    Thus, $\max_{0\leq s \leq 1}\dc(\vec{x}_0, \vec{x}_1, P^n, s) \leq n\dc(P_0, P_1, s^*)$, and substitution into \eqref{eq:dpi} gives 
\eqref{eq:opp_p}.
\item Case 3: If $\dc(\vec{x}_0, \vec{x}_1, P^n, s)$ is weakly skewed and $\dc(P_0, P_1, s)$ is balanced, then 
    \begin{align}
        &\max_{0\leq s \leq 1}\dc(\vec{x}_0, \vec{x}_1, P^n, s) \nn \\
        &~~= \max_{s \in [0,s^*] \cup [1-s^*,1]}\dc(\vec{x}_0, \vec{x}_1, P^n, s) \nn \\ &\stackrel{\text{Lem.}~\ref{lem:iid}}{\leq} n \cdot \max_{s \in [0,s^*] \cup [1-s^*,1]} \dc(0, 1, P, s).
    \end{align}
Since $\dc(P_0, P_1, s)$ is balanced, by Fact \ref{fact:monotone}, we have
\begin{align}
    &\max_{t \in [0,s^*] \cup [1-s^*,1]} \dc(0, 1, P, s^*) \nn \\
    &\quad\,\, \leq \max (\dc(0,1,P,s^*), \dc(0,1,P,1-s^*)) \nn \\ &\stackrel{\text{Assump.~\ref{assump:s_star}}}{=} \dc(P_0, P_1, s^*)
\end{align}
so that $\max_{0\leq s \leq 1}\dc(\vec{x}_0, \vec{x}_1, P^n, s) \leq n \dc(P_0, P_1, s^*)$, and \eqref{eq:opp_p} holds.
\end{itemize}
\end{proof}

To complement Lemma \ref{lem:opp_p}, we also derive a similar statement for the channel $Q$.  Recall that $\w_1(\cdot)$ is defined leading up to \eqref{eq:w1_def}, and that $\w_1(\emptyset)$ corresponds to $k=0$ therein. 

\begin{lemma}
Let $\w_1^c(\emptyset)$ be the complement string of $\w_1(\emptyset)$. If $\dc(\w_1^c(\emptyset), \w_1(\emptyset), Q^n, s)$ and $\dc(Q_0, Q_1, s)$ are of weakly opposite type, then
    \begin{equation}
        -\log (p_{e,0} + p_{e,1}) \leq n \dc(Q_0, Q_1, s^*) + \sqrt {2n} \log \frac{1}{p_{\min}} + \log 4.
        \label{eq:opp_q}
    \end{equation}
    \label{lem:opp_q}
\end{lemma}
\begin{proof}
    Combining Corollary \ref{cor:2codewords} with \eqref{eq:w1pe1} and \eqref{eq:w0pe0} (with $k=0$), we have
    \begin{align}
        -\log (p_{e,0} + p_{e,1}) &\leq \max_{0\leq s \leq 1}\dc(\w_0(\emptyset), \w_1(\emptyset), Q^n, s) \nn \\ &\qquad + \sqrt {2n} \log \frac{1}{p_{\min}} + \log 4. \label{eq:opp_q_init}
    \end{align}
    Observe that for all $w,w'$, we have
    \begin{equation}
        \dc(w',w,Q,s) \leq \dc(1-w,w,Q,s) \label{eq:ww}
    \end{equation}
    where $1-w$ is the complement of $w \in \{0,1\}$. This can be checked by considering the cases $w=w'$ (in which the LHS is 0), or $w\neq w'$ (in which $1-w=w'$). 
    Applying \eqref{eq:ww} symbol-by-symbol, taking the sum over $i=1,\dotsc,n$, and recalling the tensorization of $\dc$ (Lemma \ref{lem:iid}), we obtain
    \begin{eqnarray}
        \dc(\w_0(\emptyset), \w_1(\emptyset), Q^n, s) \leq \dc(\w_1^c(\emptyset), \w_1(\emptyset), Q^n, s).
    \end{eqnarray}
    Thus, \eqref{eq:opp_q_init} can be weakened to
    \begin{align}
        -\log (p_{e,0} + p_{e,1}) &\leq \max_{0\leq s \leq 1}\dc(\w_1^c(\emptyset), \w_1(\emptyset), Q^n, s) \nn \\ &\qquad + \sqrt {2n} \log \frac{1}{p_{\min}} + \log 4.
    \end{align}
    It remains to show that
    \begin{equation}
        \max_{0\leq s \leq 1}\dc(\w_1^c(\emptyset), \w_1(\emptyset), Q^n, s)\leq n \dc(Q_0, Q_1, s^*).
        \label{eq:w1ne}
    \end{equation}
    This follows from an identical argument to that of Lemma \ref{lem:opp_p}, with $Q$ in place of $P$.
\end{proof}

We are now ready to prove Lemma \ref{lem:opp_assump}.
\begin{proof}[Proof of Lemma \ref{lem:opp_assump}]
    We proceed with a sequence of observations: 
    \begin{itemize}
    \item If $E^* = E(P)$, then we have
    \begin{align}
        \max_{0\leq s \leq 1} \dc(\vec{x}_0, \vec{x}_1, P^n, s) &\stackrel{\text{Lem.}~\ref{lem:iid}}{\leq} n \max_{0\leq s \leq 1} \dc(P_0, P_1, s) \nn \\ &~~ = n \cdot E(P) = nE^*,
    \end{align}
    and thus \eqref{eq:dpi} (which does not rely on any of the conditions of Lemma \ref{lem:opp_p}) gives
    \begin{equation}
        -\log(p_{e,0} + p_{e,1}) \leq nE^* + \sqrt{2n} \log \frac{1}{p_{\min}} + \log 4.
    \end{equation}
    A similar analysis holds if $E^* = E(Q)$, so we can proceed by assuming $ E^* \neq E(P)$ and $E^* \neq E(Q)$.
    \item Since $ E^* \neq E(P)$ and $E^* \neq E(Q)$, by Lemma \ref{lem:dcp}, $\dc(P_0, P_1, s)$ and $\dc(Q_0, Q_1, s)$ are of strictly opposite type, and we have $E^* = \dc(P_0, P_1, s^*) = \dc(Q_0, Q_1, s^*)$ and $s^* \neq 0$.
    \item Using Lemma \ref{lem:opp_p} and \eqref{eq:perr_opp}, along with the fact that $E^* = \dc(P_0, P_1, s^*)$, we find that 
    $\dc(\vec{x}_0, \vec{x}_1, P^n,s)$ and $\dc(P_0, P_1, s)$ are of the same non-neutral type.
    \item Similarly, using Lemma \ref{lem:opp_q} and \eqref{eq:perr_opp}, along with the fact that $E^* = \dc(Q_0, Q_1, s^*)$, we find that  $\dc(\w_1^c(\emptyset), \w_1(\emptyset), Q^n, s)$ and $\dc(Q_0, Q_1, s)$ are of the same non-neutral type.

    \item The above three points also allow us to conclude that $\dc(\vec{x}_0, \vec{x}_1, P^n, s)$ and $\dc(\w_1^c(\emptyset), \w_1(\emptyset), Q^n, s)$ are of strictly opposite type.
    \item Among $\dc(\vec{x}_0, \vec{x}_1, P^n, s)$ and $\dc(\w_1^c(\emptyset), \w_1(\emptyset), Q^n, s)$, let $g_1(s)$ be the one which is balanced and $g_2(s)$ be the one which is skewed.  Since $g_2$ is skewed, we know that $g_2'(s^*)$ and $g_2'(1-s^*)$ have the same sign, while while $g_1$ being balanced implies that $g_1'(s^*)>0$ and $g_1'(1-s^*)<0$. Therefore, for one choice of $t \in \{s^*, 1-s^*\}$, $g_1'(t)$ and $g_2'(t)$ must have opposite sign, as desired.  Regarding the remaining conditions, the condition $t \notin \{0,1\}$ follows directly from $s^*\neq 0$ and $s^* \le \frac{1}{2}$, and \eqref{eq:maxpq} follows directly from $E^* = \dc(P_0, P_1, s^*) = \dc(Q_0, Q_1, s^*)$.
    \end{itemize}
\end{proof}

\section{Completion of the Proof of Theorem \ref{thm:binary_optimal_exact}}
\label{sec:main}
Obsere that if the condition \eqref{eq:perr_opp} in Lemma \ref{lem:opp_assump} fails, then Theorem \ref{thm:binary_optimal_exact} is immediately true. We thus proceed under the assumptions of Lemma \ref{lem:opp_assump}, which we make explicit as follows.

\begin{assump}
    Letting $\w_1^c(\emptyset)$ denote the bit-wise complement string of $\w_1(\emptyset)$, 
    there exists $t \in \{s^*,1-s^*\}$ (with $s^* \in \big(0,\frac{1}{2}\big]$) such that $\dc'(\vec{x}_0, \vec{x}_1, P^n, t)$ and $\dc'(\w_1^c(\emptyset), \w_1(\emptyset), Q^n, t)$ are of strictly opposite sign and
    \begin{align}
    &\max(\dc(P_0, P_1,t), \dc(P_0, P_1,1-t)) \nn \\ &= \max(\dc(Q_0, Q_1,t), \dc(Q_0, Q_1,1-t)) = E^*.
    \label{eq:dcpe}
    \end{align}
    \label{assump:opp}
\end{assump}

Before proceeding, the analysis in this section is briefly outlined as follows:
\begin{itemize}
    \item In Lemma \ref{lem:l_count}, we use Lemma \ref{lem:transition_bound} to show that at least one of two upper bounds must hold -- one on $-\log p_{e,0}(\yk)$ and one on $-\log p_{e,1}(\yk)$ (see \eqref{eq:pe0_def}--\eqref{eq:pe1_def}) -- expressed in terms of the number of 1s in ``best-case subsequences'' sent by the relay after receiving $\yk$ or $y_{1 \dotsc k+1}$.
    \item In Lemma \ref{lem:perr_likelihood}, we use Lemma \ref{lem:l_count} to identify three conditions such that at least one of them must hold for each $\yn$.  The first two of these conditions are the ones of primary interest, ensuring that certain partial log-likelihood ratios (for $P$, under $\Theta=0$ vs.~$\Theta=1$) associated with the first $k$ symbols (for some $0 \le k \le n$) are sufficiently ``well-behaved''.  The third concerns the sign of $\dc'(\vec{x}_0, \vec{x}_1, P^n, t)$ and the sign of a ``full'' (length-$n$) log-likleihood ratio; this condition is ``less desirable'' for proving a converse, but we will eventually show that its probability is sufficiently small in a suitably-defined sense.
    \item In Lemma \ref{lem:prefix_free}, we use Lemma \ref{lem:perr_likelihood} to identify a prefix-free set $A$ of strings over the output alphabet of $P$ that each satisfy one of the three preceding conditions, and such that the third condition never holds for strings of length less than $n$.
    \item In Lemma \ref{lem:pe01_lb}, we decompose the conditional error probabilities $p_{e,0}$ and $p_{e,1}$ into averages over $A$, with the prefix-free property ensuring that any $\yn$ has at most one prefix in $A$.
    \item Next, we lower bound $p_{e,0} + p_{e,1}$ by only summing over prefixes satisfying the first two of the three conditions mentioned above, and we find that we are left with $e^{-nE^* + o(n)}$ along with a multiplicative term equaling one minus the probability of the third condition holding under a suitably-defined tilted distribution, which we characterize in Lemma \ref{lem:tilted_distr}.
    \item In Lemma \ref{lem:finish}, we show that the probability just mentioned is at most $\frac{1}{2}$, which completes the proof.
\end{itemize}

Proceeding with the above outline, we first state a useful consequence of Lemma \ref{lem:transition_bound}; recall that $\w_0(\cdot)$ and $\w_1(\cdot)$ are defined following \eqref{eq:w_def}, and represent ``best-case'' sequences (given $\Theta=0$ and $\Theta=1$) following the first $k$ symbols received at the relay.  In addition, $p_{e,0}(\yk)$ and $p_{e,1}(\yk)$ are defined in \eqref{eq:pe0_def}--\eqref{eq:pe1_def}.

\begin{lemma}
    Fix a sequence $\yn$ with $P^n(\yn|\vec{x}_0)>0$ and $P^n(\yn|\vec{x}_1)>0$, and define
    \begin{equation}
        \ell_k = \hbox{number of 1s in the last $n-k$ symbols of $\w_1(\yk)$}.
        \label{eq:lk}
    \end{equation}
    For any non-negative integer $k \le n$, and any $\ell'$ (which need not be an integer) between $\ell_{k+1}$ and $\ell_k$, one of the following must hold: 
    \begin{align} 
        -\log p_{e,0}(\yk)
        &\leq \ell'\dt0 +  (n-k-\ell')\dqt1 \nn \\ & \quad + (\sqrt{2n}+1) \log \frac{1}{p_{\min}} + \log 4 \\
        -\log p_{e,1}(\yk)
        &\leq \ell'\dt1 +  (n-k-\ell')\dqt0 \nn \\ & \quad +  (\sqrt{2n}+2) \log \frac{1}{p_{\min}} + \log 4.
    \end{align}
    \label{lem:l_count}
\end{lemma}
\begin{proof}
    By definition, the first $k$ symbols of $\w_0(\yk), \w_1(\yk), \w_1(y_{1\ldots k+1})$ are identical. Therefore, any optimal hypothesis test distinguishing these strings will ignore the first $k$ symbols. Note that the number of 1s in the last $n-k$ symbols of  $\w_1(y_{1\ldots k+1})$ can be either $\ell_{k+1}$ or $\ell_{k+1}+1$, (depending on the $(k+1)$-th character).

The idea is to use Lemma \ref{lem:transition_bound} with $\w_0(\yk)$ in place of $\w_0$, $\w_1(\yk)$ in place of $\w_1$, $\w_1(y_{1\ldots k+1})$ in place of $\w'_1$, and $\ell'$ in place of $\ell$. Using \eqref{eq:w1pe1} and \eqref{eq:w0pe0}, the LHS of \eqref{eq:transition_bound1} will be replaced with 
\begin{equation}
    \max(p_{e,1}(\yk), p_{e,1}(y_{1\ldots k+1})),
\end{equation}
but we easily attain from Lemma \ref{lem:pmin_bound} that
\begin{equation}
    \max(p_{e,1}(\yk), p_{e,1}(y_{1\ldots k+1})) \leq p_{e,1}(\yk)/p_{\min}. \label{eq:pe1_min}
\end{equation}
If the number of 1s in the last $n-k$ symbols of  $\w_1(y_{1\ldots k+1})$ is $\ell_{k+1}$, or if $\ell'$ is between $\ell_{k+1}+1$ and $\ell_k$, we can directly apply Lemma \ref{lem:transition_bound} to the last $n-k$ symbols of the strings $\w_0(\yk)$, $\w_1(\yk)$, $\w_1(y_{1\ldots k+1})$ and using $\ell'$ in place of $\ell$. 
However, it is possible that there are actually $\ell_{k+1}+1$ 1s and $\ell'$ is not between $\ell_{k+1}+1$ and $\ell_k$. In this case, $\ell'$ must be between $\ell_k$ and $\ell_{k+1}$, with $0\leq \ell' - \ell_{k+1} < 1$.

Then, we can apply Lemma \ref{lem:transition_bound} with $\ell_{k+1}$ in place of $\ell$ to conclude that either
\begin{align}
    &-\log p_{e,0}(\yk) \nn \\ 
    &~~\leq \ell_{k+1} \dt0 + (n-k-\ell_{k+1})\dt1 \nn \\
        & \qquad +\sqrt{2n}\log\frac{1}{p_{\min}} + \log 4 \nonumber \\
     &~~=\ell' \dt0 + (n-k-\ell')\dt1 +\sqrt{2n}\log\frac{1}{p_{\min}} \nn \\
        & \qquad  + \log 4 + (\ell_{k+1}-\ell')(\dt0-\dt1) \nonumber \\
     &~~\stackrel{\eqref{eq:dq10}, \eqref{eq:dc_d1}}{\leq} \ell' \dt0 + (n-k-\ell')\dt1 \nn \\ &\qquad +(\sqrt{2n}+1)\log\frac{1}{p_{\min}} + \log 4
\end{align}
or (recalling \eqref{eq:pe1_min})
\begin{align}
    & -\log (p_{e,1}(\yk)/p_{\min})  \nn \\ 
    &\leq \ell_{k+1} \dt1 + (n-k-\ell_{k+1})\dt0 \nn \\ &\qquad +\sqrt{2n}\log\frac{1}{p_{\min}} + \log 4 \nonumber \\
     &=\ell' \dt1 + (n-k-\ell')\dt0 +\sqrt{2n}\log\frac{1}{p_{\min}} \nn \\ &\qquad  + \log 4 + (\ell_{k+1}-\ell')(\dt1-\dt0) \nonumber \\
     &\stackrel{\eqref{eq:dq10}, \eqref{eq:dc_d1}}{\leq} \ell' \dt1 + (n-k-\ell')\dt0 \nn \\ &\qquad+(\sqrt{2n}+1)\log\frac{1}{p_{\min}} + \log 4.
\end{align}
Combining these inequalities with $-\log p_{e,1}(\yk) = -\log (p_{e,1}(\yk)/p_{\min}) + \log \frac{1}{p_{\min}}$ completes the proof of Lemma \ref{lem:l_count}.
\end{proof}

Using Lemma \ref{lem:l_count} as a build block, we can reduce our analysis to 3 cases, two of which come down to analyzing log-likelihood ratios, and the third of which states that two $n$-letter quantities have differing signs.  As noted above, the first two cases are the ones of primary interest, and the third is ``less desirable'' for proving a converse, but will eventually be shown to have sufficiently low probability in a sense to be defined later.

\begin{lemma}
    For any $\yn$ satisfying $P^n(\yn|\vec{x}_0)>0$ and $P^n(\yn|\vec{x}_1)>0$, at least one of the following three statements is true:
    \begin{align}
        &-\log p_{e,0}(\yk) - t \sum_{i=1}^k \log \frac{P(y_i|\vec{x}_0^{(i)})}{P(y_i|\vec{x}_1^{(i)})} \nn \\ &\qquad \leq E^* \cdot (n-k)  +(\sqrt{2n}+3)\log \frac{1}{p_{\min}} + \log 4  \nn \\ &\hspace*{4.5cm} \hbox{ for some $0\leq k \leq n$}
        \label{eq:typical_bound0}
    \end{align}
    \begin{align}
        &-\log p_{e,1}(\yk) + (1-t) \sum_{i=1}^k \log \frac{P(y_i|\vec{x}_0^{(i)})}{P(y_i|\vec{x}_1^{(i)})} \nn \\
        &\qquad \leq E^* \cdot (n-k)  +(\sqrt{2n}+4)\log \frac{1}{p_{\min}} + \log 4 \nn \\ &\hspace*{4.5cm} \hbox{ for some $0\leq k \leq n$}
        \label{eq:typical_bound1}
    \end{align}
    \begin{equation}
        \left(\sum_{i=1}^n \log \frac{P(y_i|\vec{x}_0^{(i)})}{P(y_i|\vec{x}_1^{(i)})}\right) \dc'(\vec{x}_0, \vec{x}_1, P^n, t)< 0.
        \label{eq:p0_more}
    \end{equation}
    \label{lem:perr_likelihood}
\end{lemma}
\begin{proof}
    In this proof, $\yn$ is fixed, and other variables may implicitly depend on $\yn$.  
    
    Recall the definition of $\ell_k$ in \eqref{eq:lk}.  For all $0\leq k\leq n$, consider the functions
    \begin{multline}
        f_0(k) = -t \sum_{i=1}^k \log \frac{P(y_i|\vec{x}_0^{(i)})}{P(y_i|\vec{x}_1^{(i)})} \\ + \ell_{k}\dt0 + (n-k-\ell_{k})\dqt1
        \label{eq:def_f0}
    \end{multline}
    and
    \begin{multline}
        f_1(k) = (1-t) \sum_{i=1}^k \log \frac{P(y_i|\vec{x}_0^{(i)})}{P(y_i|\vec{x}_1^{(i)})} \\ + \ell_{k}\dt1 + (n-k-\ell_{k})\dqt0.
        \label{eq:def_f1}
    \end{multline}
    In addition, recalling that $\w_1^c(\emptyset)$ denotes the element-wise complement of $\w_1(\emptyset)$, we observe that
    \begin{align}
        &\dc'(\w_1^c(\emptyset), \w_1(\emptyset), Q^n, t) \nn \\ &\stackrel{\text{Lem.}~\ref{lem:iid}}{=} \sum_{i=1}^n \dc'(\w_1^{c(i)}(\emptyset), \w_1^{(i)}(\emptyset), Q^n, t)\\
        &~\,= \sum_{\w_1^{(i)}=1} \dc'(0,1,Q^n,t) + \sum_{\w_1^{(i)}=0} \dc'(1,0,Q^n,t)\\
        &~\,\stackrel{\eqref{eq:reverse1s}}{=} \sum_{\w_1^{(i)}=1} \dc'(0,1,Q^n,t) - \sum_{\w_1^{(i)}=0} \dc'(0,1,Q^n,1-t)\\
        &~\stackrel{\eqref{eq:lk}}{=} \ell_0 \dc'(0,1,Q^n,t) - (n-\ell_0)\dc'(0,1,Q^n,1-t)\\
        &~\,\stackrel{\eqref{eq:dq10}}{=} \ell_0 (\dt1 - \dt0) \nn \\
            &\qquad - (n-\ell_0)(\dqt1 - \dqt0)\\
        &~\,= f_1(0) - f_0(0). \label{eq:f1_minus_f0}
    \end{align}
    Now, by Assumption \ref{assump:opp}, $\dc'(\w_1^c(\emptyset), \w_1(\emptyset), Q^n, t)$ and $\dc'(\vec{x}_0, \vec{x}_1, P^n, t)$ have strictly opposite sign. Therefore, $f_1(0)-f_0(0)$ has strictly opposite sign with $\dc'(\vec{x}_0, \vec{x}_1, P^n, t)$.
    Moreover, substituting $k=n$ and $\ell_n=0$ in \eqref{eq:def_f0}--\eqref{eq:def_f1} gives
    \begin{equation}
        f_1(n)-f_0(n) = \sum_{i=1}^n \log \frac{P(y_i|\vec{x}_0^{(i)})}{P(y_i|\vec{x}_1^{(i)})}.
    \end{equation}
    If $f_1(0)-f_0(0)$ and $f_1(n)-f_0(n)$ are both positive or both negative, then by Assumption \ref{assump:opp} and \eqref{eq:f1_minus_f0}, $\dc'(\vec{x}_0, \vec{x}_1, P^n, t)$ and $f_1(n)-f_0(n)$ have opposite sign and thus \eqref{eq:p0_more} holds.
    
    Otherwise, there exists some integer $k$ such that $f_0(k)-f_1(k)$ and $f_0(k+1)-f_1(k+1)$ have opposite sign. 
    For this $k$ value, we define
    \begin{multline}
        f_0(k,\ell') = -t \sum_{i=1}^k \log \frac{P(y_i|\vec{x}_0^{(i)})}{P(y_i|\vec{x}_1^{(i)})}  + \ell'\dt0  \\ + (n-k-\ell')\dqt1 \label{eq:f0_def}
    \end{multline}
    and
    \begin{multline}
        f_1(k,\ell') = (1-t) \sum_{i=1}^k \log \frac{P(y_i|\vec{x}_0^{(i)})}{P(y_i|\vec{x}_1^{(i)})}  + \ell'\dt1 \\ + (n-k-\ell')\dqt0, \label{eq:f1_def}
    \end{multline}
    noting that $f_0(k, \ell_k) = f_0(k)$ and $f_1(k,\ell_k) = f_1(k)$.

    We claim that there exists $\ell'$ between $\ell_k$ and $\ell_{k+1}$ (possibly including the endpoints) such that
    \begin{equation}
        |f_0(k, \ell') - f_1(k, \ell')| \leq 2 \log \frac{1}{p_{\min}}
        \label{eq:l_intermediate}
    \end{equation}
    To see this, observe that when $\ell'=\ell_k$, 
    \begin{equation}
        f_0(k, \ell') - f_1(k, \ell') = f_0(k) - f_1(k)
    \end{equation}
    and when $\ell' = \ell_{k+1}$,
    \begin{multline}
        f_0(k, \ell') - f_1(k, \ell') = f_0(k+1) - f_1(k+1) -  \log \frac{P(y_{k+1}|\vec{x}_0^{(i)})} {P(y_{k+1}|\vec{x}_1^{(i)})} \\ \hspace*{-5ex} + \dqt1 - \dqt0
        \label{eq:small_diff}
    \end{multline}
    If $f_0(k,\ell_{k})-f_1(k,\ell_{k})$ and $f_0(k,\ell_{k+1}) - f_1(k, \ell_{k+1})$ are of opposite sign, then by continuity, we can find some $\ell'$ such that $f_0(k,\ell')-f_1(k,\ell')=0$ and $\ell'$ between $\ell_{k}, \ell_{k+1}$ so that \eqref{eq:l_intermediate} is satisfied. 
    Otherwise, because $f_0(k)-f_1(k)$ and $f_0(k+1)-f_1(k+1)$ have opposite sign, $f_0(k,\ell_{k+1})-f_1(k,\ell_{k+1})$ and $f_0(k+1)-f_1(k+1)$ also have opposite sign. By \eqref{eq:small_diff}, they differ in absolute value by at most
    \begin{align}
        &\left|\log \frac{P(y_{k+1}|\vec{x}_0^{(i)})} {P(y_{k+1}|\vec{x}_1^{(i)})}\right| + \big|\dqt1 - \dqt0\big| \nn \\ &~~~ \stackrel{\eqref{eq:dq10}}{\leq} \log \frac{1}{p_{\min}} + |\dc'(Q_0, Q_1, t)| \stackrel{\eqref{eq:dc_d1}}{\leq} 2 \log \frac{1}{p_{\min}}.
    \end{align}
    Combining the fact of differing signs with this upper bound, we obtain
    \begin{equation}
        |f_0(k,\ell_{k+1})-f_1(k,\ell_{k+1})| \leq 2 \log \frac{1}{p_{\min}}
    \end{equation}
    and again \eqref{eq:l_intermediate} holds.

    We note from Lemma \ref{lem:dc_kl} and Assumption \ref{assump:opp} that
    \begin{equation}
        (1-t)\dt0 + t\dt1 = \dc(0,1,Q,t) \leq E^*, \label{eq:E_ub1}
    \end{equation}
    and similarly
    \begin{equation}
        (1-t)\dqt1 + t\dqt0 = \dc(0,1,Q,1-t) \leq E^*. \label{eq:E_ub2}
    \end{equation}
    We now add $(1-t) \times \eqref{eq:f0_def}$ to $t \times  \eqref{eq:f1_def}$ and substitute \eqref{eq:E_ub1}--\eqref{eq:E_ub2} to obtain
    \begin{equation}
        (1-t)f_0(k,\ell') + tf_1(k,\ell') \leq (n-k)E^*,
        \label{eq:ne_bound}
    \end{equation}
    and combining this with \eqref{eq:l_intermediate}, we can conclude that
    \begin{align}
        & f_0(k,\ell') \nn \\ &= (1-t)f_0(k,\ell') + tf_1(k,\ell') + t(f_0(k,\ell') - f_1(k, \ell')) \nn \\ &\leq (n-k)E^* + 2\log \frac{1}{p_{\min}}.
        \label{eq:f0_bound}
    \end{align}
    By a similar argument, we also have
    \begin{equation}
        f_1(k,\ell') \leq (n-k)E^* + 2\log \frac{1}{p_{\min}}.
        \label{eq:f1_bound}
    \end{equation}
    
    Next, since $\ell'$ is between $\ell_k$ and $\ell_{k+1}$, we may apply Lemma \ref{lem:l_count} to conclude that at least one of the following must hold:
    \begin{align}
        -\log p_{e,0}(\yk) &\leq  \ell'\dt0 + (n-k-\ell') \dqt1 \nn \\ &\quad +(\sqrt{2n}+1)\log \frac{1}{p_{\min}} + \log 4
        \label{eq:alt_l0} \\ 
        -\log p_{e,1}(\yk) &\leq   \ell'\dt1 + (n-k-\ell') \dqt0 \nn \\ &\quad +(\sqrt{2n}+2)\log \frac{1}{p_{\min}} + \log 4.
        \label{eq:alt_l1}
    \end{align}
    In the case that \eqref{eq:alt_l0} holds, we can combine it with \eqref{eq:f0_bound} to obtain
    \begin{align}
       & -\log p_{e,0}(\yk) -t \sum_{i=1}^k \log \frac{P(y_i|\vec{x}_0^{(i)})}{P(y_i|\vec{x}_1^{(i)})}\\
        &\stackrel{\eqref{eq:alt_l0}}{\leq} \ell'\dt0 + (n-k-\ell') \dqt1 \nn \\
            &\qquad +(\sqrt{2n}+1)\log \frac{1}{p_{\min}} + \log 4 -t \sum_{i=1}^k \log \frac{P(y_i|\vec{x}_0^{(i)})}{P(y_i|\vec{x}_1^{(i)})}\\
        & \stackrel{\eqref{eq:def_f0}}{=} f_0(k,\ell') +(\sqrt{2n}+1)\log \frac{1}{p_{\min}} + \log 4\\
        &\stackrel{\eqref{eq:f0_bound}}{\leq} (n-k)E^* +(\sqrt{2n}+3)\log \frac{1}{p_{\min}} + \log 4.
    \end{align}
    Similarly, when the condition \eqref{eq:alt_l1} holds, we can combine it with \eqref{eq:f1_bound} to obtain
    \begin{align}
       & -\log p_{e,1}(\yk) +(1-t) \sum_{i=1}^k \log \frac{P(y_i|\vec{x}_0^{(i)})}{P(y_i|\vec{x}_1^{(i)})}\\
        &\stackrel{\eqref{eq:alt_l1}}{\leq}  \ell'\dt1 + (n-k-\ell') \dqt0 + \log 4 \nn \\ 
            &\qquad + (\sqrt{2n}+2)\log \frac{1}{p_{\min}} +(1-t) \sum_{i=1}^k \log \frac{P(y_i|\vec{x}_0^{(i)})}{P(y_i|\vec{x}_1^{(i)})}\\
        & \stackrel{\eqref{eq:def_f1}}{=} f_1(k,\ell') + (\sqrt{2n}+2)\log \frac{1}{p_{\min}} + \log 4\\
        &\stackrel{\eqref{eq:f1_bound}}{\leq} (n-k)E^* +(\sqrt{2n}+4)\log \frac{1}{p_{\min}} + \log 4.
    \end{align}
    This completes the proof. 
\end{proof}

Next, using the fact that at least one of \eqref{eq:typical_bound0}--\eqref{eq:p0_more} holds for all $\yn$ (subject to the positive probability constraints), we able to identify a prefix-free subset of $\calY^n$ associated with these three conditions.  Notably, the ``less desirable'' condition \eqref{eq:p0_more} is restricted to the case that $k=n$.

\begin{lemma}
    There exists a set of strings $A$ over the output alphabet of $P$ (these strings can be of different lengths), such that
    \begin{itemize}
        \item For each $\yk \in A$ with $k<n$, at least one of \eqref{eq:typical_bound0} or \eqref{eq:typical_bound1} holds.
        \item For each $\yk \in A$ with $k=n$, at least one of \eqref{eq:typical_bound0}, \eqref{eq:typical_bound1}, or \eqref{eq:p0_more} holds
        \item For any $\yn$ with $P(\yn|\Theta=0)>0$ and $P(\yn|\Theta=1)>0$, there exists some $k$ such that the sub-string $\yk$ lies in $A$.
        \item $A$ is prefix-free (i.e. for any $\yk \in A$ and all $ k'<k$, $(y_1,\ldots, y_{k'}) \notin A$).
    \end{itemize}
    \label{lem:prefix_free}
\end{lemma}
\begin{proof}
    Define $A$ as follows:
    
    For each $\yn$ such that $P(\yn|\Theta=0)>0$ and $P(\yn|\Theta=1)>0$:
    \begin{itemize}
        \item If there is some $k$ such that either \eqref{eq:typical_bound0} or \eqref{eq:typical_bound1} holds, we pick the smallest such $k$ and add $\yk$ to $A$
        \item If \eqref{eq:typical_bound0} and \eqref{eq:typical_bound1} are false for all $k$, in which \eqref{eq:p0_more} must be true by Lemma \ref{lem:perr_likelihood}, and we add $\yn$ to $A$.
    \end{itemize}
    If $\yk \in A$ and some prefix $(y_1, y_2,\ldots, y_{k'})$ were also in $A$ for some $k'<k$, then we would not have included $\yk$ as it requires $k$ to be the minimal such value. This shows that $A$ is prefix-free. The remaining properties are immediate by construction.
\end{proof}

For the remainder of this section, we let $A$ be a set satisfying the conditions of Lemma \ref{lem:prefix_free}.  The following error probability lower bounds readily follow from the prefix-free property of $A$.

\begin{lemma} \label{lem:pe01_lb}
The overall conditional error probabilities (given $\Theta=0$ or $\Theta=1$) satisfy
    \begin{equation}
        p_{e,0} \geq \sum_{\vec{y} \in A} P(\vec{y}|\Theta=0)p_{e,0}(\vec{y}) \label{eq:pe0_lb}
    \end{equation}
    and
    \begin{equation}
        p_{e,1} \geq \sum_{\vec{y} \in A} P(\vec{y}|\Theta=1)p_{e,1}(\vec{y}),
    \end{equation}
    where $P(\vec{y}|\Theta=0)$ refers to the probability that the first $|\vec{y}|$ symbols observed by the relay match $\vec{y}$, and $p_{e,0}(\vec{y})$ and $p_{e,1}(\vec{y})$ are defined in \eqref{eq:pe0_def}--\eqref{eq:pe1_def}.
\end{lemma}
\begin{proof}
    We focus on $\Theta = 0$; the argument for $\Theta = 1$ is identical. Observe that for all $\vec{y} \in A$, 
    \begin{align}
         &P(\vec{y}|\Theta=0)p_{e,0}(\vec{y}) \nn \\ &= \p(\hat{\Theta}=1 \hbox{ and } \vec{y} \hbox{ appears as a prefix of }\yn|\Theta=0).
    \end{align}
For each $\vec{y}\in A$, consider the event given by
    \begin{equation}
        \vec{y} \hbox{ appears as a prefix of $\yn$}. \label{eq:prefix_events}
    \end{equation}
    By the above properties of $A$, each $\yn$ has at most one prefix in $A$ (note that if $P(\yn|\Theta=0)=0$ or $P(\yn|\Theta=1)=0$, then $\yn$ may not have a prefix in $A$), so among the sub-strings $\{\yk\}_{k=1}^n$ associated with a given $\yn$, at most one of them can satisfy \eqref{eq:prefix_events}.  Using these findings and the law of total probability, we obtain the desired result \eqref{eq:pe0_lb}.
\end{proof}

In the following, we partition $A$ into disjoint subsets $A_0, A_1, A_2$:
\begin{itemize}
    \item $A_0$ consists of the strings such that \eqref{eq:typical_bound0} holds;
    \item $A_1$  consists of the strings such that \eqref{eq:typical_bound0} is false and \eqref{eq:typical_bound1} holds;
    \item $A_2$ consists of the strings of length $n$ such that \eqref{eq:typical_bound0} and \eqref{eq:typical_bound1} are false, but \eqref{eq:p0_more} holds.
\end{itemize}
For each $\yk \in A_0$, \eqref{eq:typical_bound0} gives us a way to bound $-\log p_{e,0}(\yk)$ directly:
\begin{align}
    & -\log \big(P(\yk|\Theta=0)p_{e,0}(\yk)\big) \nn \\ &=
    \left(-\sum_{i=1}^k \log P(y_i|\vec{x}_0^{(i)})\right) - \log p_{e,0}(\yk)\\
    &= -\log p_{e,0}(\yk) - t \sum_{i=1}^k \log \frac{P(y_i|\vec{x}_0^{(i)})}{P(y_i|\vec{x}_1^{(i)})} \nn \\
        &~~~ - (1-t) \sum_{i=1}^k \log {P(y_i|\vec{x}_0^{(i)})} - t \sum_{i=1}^k \log {P(y_i|\vec{x}_1^{(i)})}\\
    &\stackrel{\eqref{eq:typical_bound0}}{\leq} (n-k)E^* - (1-t) \sum_{i=1}^k \log {P(y_i|\vec{x}_0^{(i)})} \nn \\ &~~~ - t \sum_{i=1}^k \log {P(y_i|\vec{x}_1^{(i)})}  +(\sqrt{2n}+3)\log \frac{1}{p_{\min}} + \log 4. \label{eq:logp0_lb}
\end{align}
Similarly, for each $(\yk) \in A_1$, we can use \eqref{eq:typical_bound1} as follows:
\begin{align}
    & -\log \big(p(\yk|\Theta=1)p_{e,1}(\yk)\big) \nn \\ &=
    \left(-\sum_{i=1}^k \log P(y_i|\vec{x}_1^{(i)})\right) - \log p_{e,1}(\yk)\\
    &= -\log p_{e,1}(\yk) + (1- t) \sum_{i=1}^k \log \frac{P(y_i|\vec{x}_0^{(i)})}{P(y_i|\vec{x}_1^{(i)})} \nn \\ 
    & ~~~
    - (1-t) \sum_{i=1}^k \log {P(y_i|\vec{x}_0^{(i)})} - t \sum_{i=1}^k \log {P(y_i|\vec{x}_1^{(i)})}\\
    &\stackrel{\eqref{eq:typical_bound1}}{\leq} (n-k)E^* - (1-t) \sum_{i=1}^k \log {P(y_i|\vec{x}_0^{(i)})} \nn \\ & ~~~ - t \sum_{i=1}^k \log {P(y_i|\vec{x}_1^{(i)})}  +(\sqrt{2n}+4)\log \frac{1}{p_{\min}} + \log 4 . \label{eq:logp1_lb}
\end{align}
We can now proceed to lower bound the total error probability; starting with Lemma \ref{lem:pe01_lb}, we have
\begin{align}
    & p_{e,0} + p_{e,1}\\
    &\ge \sum_{\vec{y} \in A} P(\vec{y}|\Theta=0)p_{e,0}(\vec{y}) + \sum_{\vec{y} \in A} P(\vec{y}|\Theta=1)p_{e,1}(\vec{y})\\
    & \geq  \sum_{\vec{y} \in A_0} P(\vec{y}|\Theta=0)p_{e,0}(\vec{y}) + \sum_{\vec{y} \in A_1} P(\vec{y}|\Theta=1)p_{e,1}(\vec{y})\\
    & \stackrel{\eqref{eq:logp0_lb}, \eqref{eq:logp1_lb}}{\geq}  \sum_{\vec{y} \in A_0 \cup A_1} e^{-(n-k)E^* - (\sqrt{2n}+4)\log \frac{1}{p_{\min}} - \log 4} \nn \\
        &\qquad\qquad \times  P(\vec{y}|\Theta=0)^{1-t} P(\vec{y}|\Theta=1)^t \label{eq:sum_a0_a1_} \\
    &\stackrel{\eqref{eq:dcpe}}{\geq} \sum_{\vec{y} \in A_0 \cup A_1} \exp\left( \sum_{i=1}^k \dc(\vec{x}_0^{(i)}, \vec{x}_1^{(i)}, P, t)\right) \nn \\ & \qquad\qquad \times P(\vec{y}|\Theta=0)^{1-t} P(\vec{y}|\Theta=1)^t \nn \\ & \qquad\qquad \times e^{-nE^*  - (\sqrt{2n}+4)\log \frac{1}{p_{\min}} - \log 4},
    \label{eq:sum_a0_a1}
\end{align}
where in \eqref{eq:sum_a0_a1_} and \eqref{eq:sum_a0_a1} we recall that $k$ depends on $\vec{y}$.

It will be useful to interpret the terms being summed in \eqref{eq:sum_a0_a1} (excluding the $e^{-nE^* - \dotsc}$ term) as a tilted probability distribution along similar lines to $P_s$ and $Q_s$ in Definition \ref{def:tilted}.  Specifically, the following lemma will be used for this purpose.

\begin{lemma} \label{lem:tilted_distr}
For each $1\leq i\leq n$, let $\Ptilde_{t}^{(i)}$ be the probability distribution defined by
\begin{align}
    \Ptilde_{t}^{(i)}(y) &= \frac{P(y_i|\vec{x}_0^{(i)})^{1-t}P(y_i|\vec{x}_1^{(i)})^{t}}{\sum_{y' \in \calY}    P(y'|\vec{x}_0^{(i)})^{1-t}P(y'|\vec{x}_1^{(i)})^{t}} \nn \\ &= P(y_i|\vec{x}_0^{(i)})^{1-t}P(y_i|\vec{x}_1^{(i)})^{t} \exp(\dc(\vec{x}_0^{(i)}, \vec{x}_1^{(i)}, P, t)). \label{eq:Ptilde_def}
\end{align}
Moreover, let $\Ptilde_t$ be the product distribution defined by $\Ptilde_t(\yn) = \prod_{i=1}^n \Ptilde_{t}^{(i)}(y_i)$;
 with a slight abuse of notation we also define $\Ptilde_t(\yk)$ similarly when $k < n$. Then, we have
    \begin{multline}
        \Ptilde_t(\yk) = \exp \left(\sum_{i=1}^k \dc(\vec{x}_0^{(i)}, \vec{x}_1^{(i)}, P, t)\right)  \\ \times P(\yk|\Theta=0)^{1-t} P(\yk|\Theta=1)^t
        \label{eq:specific_string}
    \end{multline}
\end{lemma}
    \begin{proof}
        Since $P$ is memoryless, we have that \eqref{eq:specific_string} factorizes into
        \begin{align}
            & \prod_{i=1}^k \Big(\exp\left(\dc(\vec{x}_0^{(i)}, \vec{x}_1^{(i)}, P, t)\right) P(y_i|x_0^{(i)})^{1-t} P(y_i|x_1^{(i)})^t\Big)\\
            &=  
            \prod_{i=1}^k P(y_i|x_0^{(i)})^{1-t} P(y_i|x_1^{(i)})^t \exp(\dc(\vec{x}_0^{(i)}, \vec{x}_1^{(i)}, P, t))
        \end{align}
        as required.
    \end{proof}
For each $\vec{y}$ such that $\Ptilde_t(\vec{y})>0$, we must have $P^n(\yn|\vec{x}_0)>0$ and $P^n(\yn|\vec{x}_1)>0$, so by the properties established previously for $A$, exactly one string in $A$ must appear as a prefix when $\yn$ is distributed over $\Ptilde_t$.  Therefore,
\begin{equation}
    \sum_{\vec{y} \in A} \Ptilde_t(\vec{y}) = 1
\end{equation}
Equation \eqref{eq:sum_a0_a1} now simplifies to
\begin{align}
    &p_{e,0} + p_{e,1} \nn \\ &\geq e^{-nE^* - (\sqrt{2n}+4)\log \frac{1}{p_{\min}} - \log 4}  \sum_{\vec{y} \in A_0 \cup A_1} \Ptilde_t(\vec{y})\\
    &= e^{-nE^* - (\sqrt{2n}+4)\log \frac{1}{p_{\min}}-\log 4}\left(1- \sum_{\vec{y} \in A_2} \Ptilde_t(\vec{y})\right),
    \label{eq:perr_final}
\end{align}
where the summation of $A_2$ arises by recalling that $A_0,A_1,A_2$ forms a disjoint partition of $A$.

With the above in place, it is sufficient to show the following in order to complete the proof of Theorem \ref{thm:binary_optimal_exact}.

\begin{lemma} 
Under the preceding definitions, it holds that
\begin{equation}
    \sum_{\vec{y} \in A_2} \Ptilde_t(\vec{y}) \leq \frac12.
    \label{eq:finish}
\end{equation}
\label{lem:finish}
\end{lemma}
\begin{proof}
Recall the definitions $p_{e,0}(\yk)$ and $p_{e,1}(\yk)$ from \eqref{eq:pe0_def}--\eqref{eq:pe1_def}.  Observe that when $k=n$, the conditioning on $\Theta$ becomes irrelevant due to the Markov chain relation $\Theta \to \yn \to \hat{\Theta}$.  Thus, for any sequence $\yn$, we have
\begin{equation}
    p_{e,0}(\yn) = \p_{\vec{z} \sim Q^n(\cdot | \w(\yn))} (\vec{z} \in D_1)
\end{equation}
and
\begin{equation}
    p_{e,1}(\yn) = \p_{\vec{z} \sim Q^n(\cdot | \w(\yn))} (\vec{z} \in D_0),
\end{equation}
where $D_0, D_1$ are decoding regions corresponding to $\Theta=0$ and $\Theta=1$ respectively.  Since $D_0$ and $D_1$ form a partition of the output space, we have $p_{e,0}(\yn) + p_{e,1}(\yn) = 1$, which implies that either $p_{e,0}(\yn) \ge \frac{1}{2}$ or $p_{e,1}(\yn) \ge \frac{1}{2}$.

By the definition of $A_2$, every element of $A_2$ has length exactly $n$, and \eqref{eq:typical_bound0}--\eqref{eq:typical_bound1} must both be false. If $p_{e,0}(\yn) \geq \frac{1}{2}$, then the negation of \eqref{eq:typical_bound0} (along with the crude bounds $t \le 1$ and $E(n-k) \ge 0$) implies
\begin{align}
    &\left| \sum_{i=1}^n \log \frac{P(y_i|\vec{x}_0^{(i)})}{P(y_i|\vec{x}_1^{(i)})}\right| \nn \\ &~~\geq (\sqrt{2n}+3)\log \frac{1}{p_{\min}} + \log 4 + \log p_{e,0}(\yn) \nn \\ &~~\geq \sqrt{2n} \log \frac{1}{p_{\min}}.
\end{align}
Similarly, if $p_{e,1}(\yn) \geq \frac{1}{2}$, then the negation of \eqref{eq:typical_bound1} implies
\begin{align}
    &\left| \sum_{i=1}^n \log \frac{P(y_i|\vec{x}_0^{(i)})}{P(y_i|\vec{x}_1^{(i)})}\right| \nn \\ &~~\geq (\sqrt{2n}+4)\log \frac{1}{p_{\min}} + \log 4 + \log p_{e,1}(\yn) \nn \\ &~~\geq \sqrt{2n} \log \frac{1}{p_{\min}}.
\end{align}
Since one of these two cases must hold, we conclude that for each $\yn \in A_2$,
\begin{equation}
    \left| \sum_{i=1}^n \log \frac{P(y_i|\vec{x}_0^{(i)})}{P(y_i|\vec{x}_1^{(i)})}\right| \geq \sqrt{2n} \log \frac{1}{p_{\min}}.
    \label{eq:large_likelihood}
\end{equation}
Let $B$ be the set of all strings of length $n$ such that \eqref{eq:p0_more} and \eqref{eq:large_likelihood} are true, noting that $A_2 \subseteq B$. Instead of doing the summation over $A_2$, it is easier to sum over $B$. Thus, we we will instead show that
\begin{equation}
    \sum_{\vec{y} \in B} \Ptilde_t(\vec{y}) \leq \frac12,
\end{equation}
or equivalently,
\begin{multline}
    \p_{\yn \sim \Ptilde_t} \Bigg(\left(\sum_{i=1}^n \log \frac{P(y_i|\vec{x}_0^{(i)})}{P(y_i|\vec{x}_1^{(i)})} \right) \dc'(\vec{x}_0, \vec{x}_1, P^n, t)< 0  \\ \hbox{ and } \left| \sum_{i=1}^n \log \frac{P(y_i|\vec{x}_0^{(i)})}{P(y_i|\vec{x}_1^{(i)})}\right| \geq \sqrt{2n} \log \frac{1}{p_{\min}}\Bigg) \leq \frac12.
    \label{eq:iid_sum}
\end{multline}

Observe that
\begin{align}
    &\e_{\yn \sim \Ptilde_t} \left(\log \frac{\pyxii0}{\pyxii1}\right) \nn \\ &\stackrel{\eqref{eq:Ptilde_def}}{=} \sum_{y \in \calY}  \pyxi0^{1-t}\pyxi1^{1-t}\exp(\dc(\vec{x}_0^{(i)}, \vec{x}_1^{(i)}, P,t)) \nn \\
        &\hspace*{4.5cm} \times \left(\log \frac{\pyxi0}{\pyxi1}\right)\\
    &\stackrel{\eqref{eq:expected_s_log}}{=} \dc'(\vec{x}_0^{(i)}, \vec{x}_1^{(i)}, P, t).
\end{align}
Under the distribution $\Ptilde_t$, the entries of $\yn$ are independent, and therefore
\begin{equation}
    \sum_{i=1}^n \log \frac{P(y_i|\vec{x}_0^{(i)})}{P(y_i|\vec{x}_1^{(i)})}
    \label{eq:iid_sum2}
\end{equation}
is a sum of independent variables with expected value
\begin{align}
    \sum_{i=1}^n \e \left(\log \frac{P(y_i|\vec{x}_0^{(i)})}{P(y_i|\vec{x}_1^{(i)})}\right) &= \sum_{i=1}^n \dc'(\vec{x}_0^{(i)}, \vec{x}_1^{(i)}, P, t) \nn \\ &\stackrel{\text{Lem.}~\ref{lem:iid}}{=} \dc'(\vec{x}_0, \vec{x}_1, P^n, t)
\end{align}
and the variance is bounded above by
\begin{align}
    &{\rm Var} \sum_{i=1}^n  \left(\log \frac{P(y_i|\vec{x}_0^{(i)})}{P(y_i|\vec{x}_1^{(i)})}\right) \nn \\ &\qquad\leq \sum_{i=1}^n \max_{y_i \,:\,\Ptilde_t^{(i)}(y_i)>0} \left( \log \frac{P(y_i|\vec{x}_0^{(i)})}{P(y_i|\vec{x}_1^{(i)})}\right)^2 \nn \\ &\qquad\leq n \left(\log \frac{1}{p_{\min}}\right)^2.
\end{align}
The event that $\vec{y} \in B$ is equivalent to \eqref{eq:iid_sum2} having an opposite sign to its expected value $\dc'(\vec{x}_0, \vec{x}_1, P^n, t)$ and having absolute value at least $\sqrt{2n}\log \frac{1}{p_{\min}}$, which means that it must deviate from its expected value by at least $\sqrt{2n}\log \frac{1}{p_{\min}}$. We can thus use Chevyshev inequality to obtain
\begin{equation}
    \sum_{\vec{y} \in B} \Ptilde_t(\vec{y}) \leq \frac{n \big(\log \frac{1}{p_{\min}}\big)^2}{(\sqrt{2n}\log \frac{1}{p_{\min}})^2} =\frac12.
\end{equation}
which completes the proof of Lemma \ref{lem:finish}, and thus also the proof of Theorem \ref{thm:binary_optimal_exact}.
\end{proof}

\section{Conclusion}
In this paper, we studied the error exponent of relaying a single bit over a tandem of channels. We proved a matching converse bound to the protocol in \cite{teachlearn} whenever the two channels have binary inputs. 

In \cite{teachlearn}, we also raised another open problem regarding channels with more than two inputs: \emph{Does there exist a pair of DMCs $P$, $Q$, such that every restriction of $Q$ to two inputs gives a strictly suboptimal learning rate?}  Since the present paper only deals with binary-input channels, this problem still remains open. The techniques used in this paper do not appear to generalize easily to channels with more than two inputs; for instance, the use of Lemma \ref{lem:nbound} requires bounding certain error probabilities by the number of positions in which the strings are 0 and 1 respectively, and it is unclear to what extent this can be extended to non-binary inputs.  At a more conceptual level, having all three inputs could, in principle, be beneficial compared to having only the ``best two'' of them; for example, a third input could be used by the relay to convey useful information such as the degree of uncertainty.

\section*{Acknowledgment}

This work was supported by the Singapore National Research Foundation (NRF) under grant number A-0008064-00-00.

\bibliographystyle{IEEEtran}
\bibliography{main.bbl}

\begin{thebibliography}{10}
\providecommand{\url}[1]{#1}
\csname url@samestyle\endcsname
\providecommand{\newblock}{\relax}
\providecommand{\bibinfo}[2]{#2}
\providecommand{\BIBentrySTDinterwordspacing}{\spaceskip=0pt\relax}
\providecommand{\BIBentryALTinterwordstretchfactor}{4}
\providecommand{\BIBentryALTinterwordspacing}{\spaceskip=\fontdimen2\font plus
\BIBentryALTinterwordstretchfactor\fontdimen3\font minus
  \fontdimen4\font\relax}
\providecommand{\BIBforeignlanguage}[2]{{%
\expandafter\ifx\csname l@#1\endcsname\relax
\typeout{** WARNING: IEEEtran.bst: No hyphenation pattern has been}%
\typeout{** loaded for the language `#1'. Using the pattern for}%
\typeout{** the default language instead.}%
\else
\language=\csname l@#1\endcsname
\fi
#2}}
\providecommand{\BIBdecl}{\relax}
\BIBdecl

\bibitem{onebit}
W.~Huleihel, Y.~Polyanskiy, and O.~Shayevitz, ``Relaying one bit across a
  tandem of binary-symmetric channels,'' \emph{IEEE International Symposium on
  Information Theory (ISIT)}, 2019.

\bibitem{jog2020teaching}
V.~{Jog} and P.~L. {Loh}, ``Teaching and learning in uncertainty,'' \emph{IEEE
  Transactions on Information Theory}, vol.~67, no.~1, pp. 598--615, 2021.

\bibitem{jadbabaie2013information}
A.~Jadbabaie, P.~Molavi, and A.~Tahbaz-Salehi, ``Information heterogeneity and
  the speed of learning in social networks,'' \emph{Columbia Business School
  Research Paper}, no. 13-28, 2013.

\bibitem{molavi2017foundations}
P.~Molavi, A.~Tahbaz-Salehi, and A.~Jadbabaie, ``Foundations of non-{B}ayesian
  social learning,'' \emph{Columbia Business School Research Paper}, no. 15-95,
  2017.

\bibitem{teachlearn}
Y.~H. Ling and J.~Scarlett, ``Optimal rates of teaching and learning under
  uncertainty,'' \emph{IEEE Transactions on Information Theory}, vol.~61,
  no.~11, pp. 7067--7080, 2021.

\bibitem{berlekampI}
C.~Shannon, R.~Gallager, and E.~Berlekamp, ``Lower bounds to error probability
  for coding on discrete memoryless channels. i,'' \emph{Information and
  Control}, vol.~10, no.~1, p. 65–103, 1967.

\bibitem{multibit}
Y.~H. Ling and J.~Scarlett, ``Multi-bit relaying over a tandem of channels,''
  \emph{IEEE Transactions on Information Theory}, 2023.

\bibitem{multirelay}
M.~Bliss, C.-C. Wang, and D.~J. Love, ``Optimal single-bit relaying strategies
  with multi-relay diversity,'' \emph{IEEE Transactions on Information Theory},
  vol.~69, no.~10, pp. 6518--6536, 2023.

\bibitem{optbiso}
C.-C. Wang and D.~J. Love, ``Optimal learning rate of sending one bit over
  arbitrary acyclic {BISO}-channel networks,'' in \emph{IEEE International
  Symposium on Information Theory (ISIT)}, 2023, pp. 1360--1365.

\bibitem{maxflow}
Y.~H. Ling and J.~Scarlett, ``Maxflow-based bounds for low-rate information
  propagation over noisy networks,'' \emph{IEEE Transactions on Information
  Theory (to appear)}, 2024.

\bibitem{infovelocity}
------, ``Simple coding techniques for many-hop relaying,'' \emph{IEEE
  Transactions on Information Theory}, vol.~68, no.~11, pp. 7043--7053, 2022.

\bibitem{packeterasure}
E.~Domanovitz, T.~Philosof, and A.~Khina, ``The information velocity of
  packet-erasure links,'' in \emph{IEEE Conference on Computer Communications
  (INFOCOM)}.\hskip 1em plus 0.5em minus 0.4em\relax IEEE, 2022, pp. 190--199.

\bibitem{schulman_1994}
S.~Rajagopalan and L.~Schulman, ``A coding theorem for distributed
  computation,'' \emph{ACM Symposium on Theory of Computing (STOC)}, 1994.

\bibitem{cover_thomas}
T.~M. Cover and J.~A. Thomas, \emph{Elements of information theory}.\hskip 1em
  plus 0.5em minus 0.4em\relax John Wiley \& Sons, Inc., 2006.

\bibitem{gamalkim}
A.~El~Gamal and Y.-H. Kim, \emph{Network information theory}.\hskip 1em plus
  0.5em minus 0.4em\relax Cambridge University Press, 2011.

\bibitem{bradford2012error}
G.~J. Bradford and J.~N. Laneman, ``Error exponents for block {M}arkov
  superposition encoding with varying decoding latency,'' in \emph{IEEE
  Information Theory Workshop (ITW)}.\hskip 1em plus 0.5em minus 0.4em\relax
  IEEE, 2012, pp. 237--241.

\bibitem{highrates}
V.~Y.~F. Tan, ``On the reliability function of the discrete memoryless relay
  channel,'' \emph{IEEE Transactions on Information Theory (ISIT)}, vol.~61,
  no.~4, pp. 1550--1573, 2015.

\bibitem{ogbe2019optimal}
D.~Ogbe, C.-C. Wang, and D.~J. Love, ``On the optimal delay amplification
  factor of multi-hop relay channels,'' in \emph{IEEE International Symposium
  on Information Theory (ISIT)}.\hskip 1em plus 0.5em minus 0.4em\relax IEEE,
  2019, pp. 2913--2917.

\end{thebibliography}
    
\end{document}